\documentclass[runningheads]{llncs}

\usepackage{graphicx}
\usepackage{wrapfig}
\usepackage{amsmath}
\usepackage{amssymb}
\usepackage{amsfonts}
\usepackage{xcolor}
\usepackage{booktabs}
\usepackage{makecell}
\usepackage{listings}
\usepackage{hyperref}
\usepackage[T1]{fontenc}
\usepackage{textcomp}
\usepackage{appendix}
\usepackage{subcaption} 
\usepackage{multirow} 
\usepackage{bussproofs}
\usepackage[algoruled,longend,linesnumbered]{algorithm2e}

\newcommand\natnum{\mathbb{N}}
\newcommand\fwc{\Phi_{W}}
\newcommand\firlis{$\fwc$}

\newcommand\sccgraph{\mathcal{S}}

\newcommand{\prettyll}[1]{\begin{lstinline}[basicstyle=\small\ttfamily]! #1 !\end{lstinline}}
\newcommand{\fnprettyll}[1]{\begin{lstinline}[basicstyle=\footnotesize\ttfamily]! #1 !\end{lstinline}}

\newcommand{\intnum}{\mathbb{Z}}

\newcommand{\firrtl}{\textrm{FIRRTL}}
\newcommand{\hifirrtl}{\textrm{HiFIRRTL}}
\newcommand{\lofirrtl}{\textrm{LoFIRRTL}}

\newcommand{\eqdef}{\stackrel{\mbox{\begin{tiny}def\end{tiny}}}{=}} 

\newcommand {\floor}[1] {\ensuremath{\lfloor #1 \rfloor}}

\newcommand{\ltrue}{{\sf true}}
\newcommand{\lfalse}{{\sf false}}

\newcommand{\cut}[1]{}

\mathchardef\mhyphen="2D 

\newcommand{\hide}[1]{}

\newcommand\dom{\mathsf{dom}}

\newcommand{\concat}{\cdot}

\newcommand{\OMIT}[1]{}

\newcommand{\seq}[1]{\ensuremath{#1}}
\newcommand{\seqq}[1]{\seq{\ifx#1\relax\else#1\fi}}

\newcommand{\ruleName}[1]{\textrm{#1}}

\newcommand{\infernone}[2][]{%
	\AxiomC{}
	\ifx#1\relax\else\LeftLabel{\ruleName{#1}}\fi
	\UnaryInfC{$#2$}
	\DisplayProof
}

\newcommand{\infer}[3][]{%
  \AxiomC{$#3$}
  \ifx#1\relax\else\LeftLabel{\ruleName{#1}}\fi
  \UnaryInfC{$#2$}
  \DisplayProof
}
\newcommand{\inferC}[4][]{%
  \AxiomC{$#4$}
  \RightLabel{$~~#2$}
  \ifx#1\relax\else\LeftLabel{\ruleName{#1}}\fi
  \UnaryInfC{$#3$}
  \DisplayProof
}
\newcommand{\inferii}[4][]{%
  \AxiomC{$#3$}
  \AxiomC{$#4$}
  \ifx#1\relax\else\LeftLabel{\ruleName{#1}}\fi
  \BinaryInfC{$#2$}
  \DisplayProof
}
\newcommand{\inferiii}[5][]{%
  \AxiomC{$#3$}
  \AxiomC{$#4$}
  \AxiomC{$#5$}
  \ifx#1\relax\else\LeftLabel{\ruleName{#1}}\fi
  \TrinaryInfC{$#2$}
  \DisplayProof
}
\newcommand{\inferiv}[6][]{%
  \AxiomC{$#3$}
  \AxiomC{$#4$}
  \AxiomC{$#5$}
  \AxiomC{$#6$}
  \ifx#1\relax\else\LeftLabel{\ruleName{#1}}\fi
  \QuaternaryInfC{$#2$}
  \DisplayProof
}
\newcommand{\inferv}[7][]{%
  \AxiomC{$#3$}
  \AxiomC{$#4$}
  \AxiomC{$#5$}
  \AxiomC{$#6$}
    \AxiomC{$#7$}
  \ifx#1\relax\else\LeftLabel{\ruleName{#1}}\fi
  \QuinaryInfC{$#2$}
  \DisplayProof
}

\newcommand\nil{{\sf nil}}

\newcommand{\firtool}{{\sf firtool}}

\newif\ifdraft\draftfalse
\ifdraft

\newcommand{\zhilin}[1]{\color{brown} {ZL: #1 :LZ} \color{black}}
\newcommand{\xiaomu}[1]{\color{blue} {XM: #1 :MX} \color{black}}
\newcommand{\david}[1]{\color{cyan} {DNJ: #1 :DNJ} \color{black}}
\newcommand{\keyin}[1]{\color{red} {KY: #1 :YK} \color{black}}
\newcommand{\fu}[1]{\color{purple} {Fu: #1} \color{black}}
\newcommand{\tl}[1]{\textcolor{magenta}{TC: #1 :CT} }
\newcommand{\revise}[1]{\textcolor{blue}{#1}}
\else
\newcommand{\zhilin}[1]{}
\newcommand{\xiaomu}[1]{}
\newcommand{\david}[1]{}
\newcommand{\keyin}[1]{}
\newcommand{\fu}[1]{}
\newcommand{\tl}[1]{}
\newcommand{\revise}[1]{#1}
\fi

\newcommand{\arxiv}[1]{#1}
\newcommand{\crv}[1]{}

\usepackage{listings}
\usepackage{color}
\definecolor{lightgray}{rgb}{.9,.9,.9}
\definecolor{darkgray}{rgb}{.4,.4,.4}
\definecolor{purple}{rgb}{0.65, 0.12, 0.82}

\lstdefinelanguage{Firrtl}{
  keywords={defname, parameter, with, reset, inst, of, is, invalid, attach, when, else, stop, printf, skip, data-type, depth, read-latency, write-latency, read-under-write, reader, writer, readwriter, old, new, undefined, UInt, SInt, mux, validif, flip, Clock, Reset, AsyncReset, Analog, Fixed, add, sub, mul, div, mod, lt, leq, gt, geq, eq, neq, dshl, dshr, and, or, xor, cat, asUInt, asSInt, asClock, cvt, neg, not, andr, orr, xorr, head, tail, pad, shl, shr, bits},
  keywordstyle=\color{blue},
  ndkeywords={circuit, module, extmodule, input, output, wire, reg, node, mem},
  ndkeywordstyle=\color{purple},
  identifierstyle=\color{black},
  sensitive=true,
  stringstyle=\color{red}\ttfamily,
  morestring=[b]',
  morestring=[b]"
}

\title{A Formally Verified Procedure for~Width~Inference in~FIRRTL\thanks{This work was supported by the Strategic Priority Research Program of the Chinese Academy of Sciences, Grant No.~XDA0320101, and partially supported by NSFC-RGC Collaborative Research Grant No.~62561160151. D.N. Jansen is supported by Beijing Natural Science Foundation Project No.~IS25071.}}

\author{
Keyin Wang\inst{1,2,3}\orcidID{0009-0004-1954-5171} \and
Xiaomu Shi\inst{1,2,3}\orcidID{0000-0001-6277-2813} \and
Jiaxiang Liu\inst{1,2,3}\orcidID{0000-0002-6725-8167} \and
Zhilin Wu\inst{1,2,3}\orcidID{0000-0003-0899-628X}\and\\
Fu Song\inst{1,2,3,4}\orcidID{0000-0002-0581-2679}\and
Taolue Chen\inst{5}\orcidID{0000-0002-5993-1665}\and\\
David N. Jansen\inst{1,2,3}\orcidID{0000-0002-6636-3301}
}
\authorrunning{K. Wang et al.}
\institute{Key Laboratory of System Software (Chinese Academy of Sciences)
\and Institute of Software, Chinese Academy of Sciences 
\and University of Chinese Academy of Sciences
\and Nanjing Institute of Software Technology
\and Birkbeck, University of London
\email{\{wangky,shixm,liujx,wuzl,songfu,dnjansen\}@ios.ac.cn} \\
\email{t.chen@bbk.ac.uk}}

\begin{document}
\maketitle

\begin{abstract}
{\firrtl} is an intermediate representation language for Register Transfer Level (RTL)  hardware designs. In {\firrtl} programs, the bit widths of many components are not specified explicitly and must be inferred during compilation. In mainstream {\firrtl} compilers, such as the official compiler {\firtool}, width inference is conducted by a compilation pass referred to as InferWidths, which may fail even for simple {\firrtl} programs. In this paper, we thoroughly investigate the width inference problem for {\firrtl} programs.
We show that, if the constraints obtained from a {\firrtl} program are satisfiable, there exists a unique least solution. Based on this result, we propose a complete procedure for solving the width inference problem. 
We implement it in the interactive theorem prover Rocq and prove its functional correctness. From the Rocq implementation, we extract an OCaml implementation, which is the first formally verified implementation of the InferWidths pass. Extensive experiments demonstrate that our approach can solve more instances than the official InferWidths pass in {\firtool}, normally with high efficiency. 
\end{abstract}

\section{Introduction}

Chisel (\underline{C}onstructing \underline{H}ardware \underline{i}n a \underline{S}cala \underline{E}mbedded \underline{L}anguage~\cite{BachrachVRLWAWA12}) is an open-source hardware description language, designed to describe digital electronics and circuits at the register-transfer level (RTL).
As an embedded domain-specific language, Chisel is powered by Scala, bringing all the power of object-oriented and functional programming to type-safe hardware design and generation.
It also comprises a library of special class definitions, predefined objects and usage conventions. 
High-level hardware generators in Chisel are programmed by manipulating circuit components using the Chisel standard library as well as Scala functions including functional and object-oriented programming features, where their interfaces are encoded by Scala types. This enables agile methodologies of highly-parameterized, modular, and reusable hardware generators that improve the productivity and robustness of hardware designs~\cite{BaoC20}.
Chisel has been successfully used to develop RISC-V processors such as Rocket Chip~\cite{RocketChip}, RISC-V BOOM \cite{Boom}, NutShell~\cite{NutShell}, and XiangShan~\cite{xiangshan,Xiangshan-github}.

{\firrtl} (\underline{F}lexible \underline{I}ntermediate \underline{R}epresentation for \underline{RTL}~\cite{IzraelevitzKLLW17,firrtl-spec}) was proposed  as an intermediate representation (IR) language for the compilation of Chisel programs, analogous to the role of LLVM IR~\cite{LLVM} for compiling C/C++ programs into low-level code.
It features first-class support for high-level constructs such as vector types, bundle types, conditional statements, connects and modules, so that Chisel programs can easily be translated into {\firrtl}.
High-level synthesis is performed on {\firrtl} by a sequence of compilation and optimization passes, finally resulting in a most restricted form that is close to synthesizable RTL Verilog, commonly called \emph{low {\firrtl}} ({\lofirrtl} for short).
For clarity, we refer to {\firrtl} with high-level constructs as \emph{high {\firrtl}} ({\hifirrtl} for short). While {\firrtl} was originally designed to compile Chisel programs, it is not tied to Chisel. Other hardware description languages can also be compiled into {\firrtl} and reuse the majority of the compilation and optimization passes in {\firrtl} compilers. For instance, {\firrtl} has been integrated into the CIRCT project~\cite{CIRCT} as one of its core dialects. The tool {\firtool}, as a part of the CIRCT project, is the current official {\firrtl} compiler and used extensively in the design of RISC-V processors.


\subsubsection*{Width inference in {\firrtl}.}  
When working with {\firrtl}, users are encouraged to manually specify (bit) widths of ports and registers. 
Unspecified widths are inferred during compilation. 
These inferred widths are 
ideally 
as small as possible depending on 
the widths of all incoming connections.

\revise{Consider the {\firrtl} program in Fig.~\ref{fig-infer-width-acyclic}, which is adapted from 
the Chisel book~\cite{chisel:book}.
The module {\tt CombWhen} has two inputs  {\tt clock} and {\tt cond}, one output {\tt out}, and a wire {\tt w}. The widths of {\tt cond} and {\tt out} are $1$ and $4$, respectively, while the width of {\tt w} is unspecified. Although {\firrtl} adopts a ``last connect semantics'', 
i.e., the value of the most recently connected expression is assigned to {\tt w},  width inference must take all connections into account to maintain their legality. So, the width of {\tt w} must be no less than those of \prettyll{UInt<1>("h0")} and \prettyll{UInt<2>("h3")},
i.e., no less than $2$. The \textsc{InferWidths} compilation pass implemented in {\firtool}, the current official {\firrtl} compiler, correctly infers this width for {\tt w}. }

 \begin{figure}[t]
 	\centering
  \begin{minipage}[t]{0.43\linewidth} 	\centering
 \begin{tabular}{cc}
\quad\quad&
\begin{lstlisting}
circuit CombWhen :
  module CombWhen :
    input clock : Clock
    input cond : UInt<1>
    output out : UInt<4>
    wire w : UInt 
    w <= UInt<1>("h0") 
    when cond :
      w <= UInt<2>("h3")
    out <= w 
\end{lstlisting}
 	  	\end{tabular}	
 	  	\caption{A {\firrtl} program with unspecified widths} 
 	  	\label{fig-infer-width-acyclic}
 	\end{minipage}
 \hfill
\begin{minipage}[t]{0.5\linewidth} 	\centering
\begin{tabular}{c}
\begin{lstlisting}
circuit A :
  module A :
    input in : UInt<4>
    input clock : Clock
    output out : UInt
    reg x : UInt, clock
    
    x <= add(tail(x,1), in)
    out <= x
\end{lstlisting}
\end{tabular}
 	\caption{A {\firrtl} program with unspecified widths and circular dependencies}
 	\label{fig-infer-width-circular}
 \end{minipage}
 	\vspace{-4mm}
\end{figure}

%
%
%

This example is fairly straightforward because there are \emph{no circular dependencies} between (circuit) components. Note that a component {\tt x} depends on another component {\tt y} if there is a connect statement \prettyll{x <= exp} such that {\tt y} occurs in the expression {\tt exp}.
Circular dependencies between components make the width inference problem more challenging.  
For instance, for the program in Fig.~\ref{fig-infer-width-circular} where the widths of {\tt x} and {\tt out} are unspecified, {\tt x} is dependent on itself, giving rise to (perhaps the simplest) circular dependency. In this case, the InferWidths pass of {\firtool} fails to infer the width for {\tt x} and throws an exception.

Intuitively, the width inference for the program in Fig.~\ref{fig-infer-width-circular} can be reduced to finding the least solution of the constraint
\begin{center}
  $w_{\tt x} \ge \max(w_{\tt x}-1, 4)+1\ \wedge w_{\tt out} \ge w_{\tt x},$  
\end{center}
where $w_{\tt x}$ and $w_{\tt out}$ are nonnegative integer variables for the widths of {\tt x} and {\tt out}, respectively. 
(According to the {\firrtl} specification~\cite{firrtl-spec}, the widths of the expressions \prettyll{tail(e1, n)} and \prettyll{add(e1, e2)}  
are $w_{{\tt e1}}-{\tt n}$ and $\max(w_{{\tt e1}}, w_{{\tt e2}})+1$.)
It is not hard to see that this constraint has a least solution, i.e.\@ $w_{{\tt x}} = w_{{\tt out}} =5$. 

\revise{
This example reveals that the current implementation of the InferWidths pass in {\firtool} is incomplete. A complete InferWidths pass is crucial:
if the inferred widths are too large, the synthesized hardware wastes area and power, and may suffer degraded performance due to increased routing and timing complexity.
Conversely, if the inferred widths are too small, arithmetic overflow or data loss can occur, leading to incorrect results.
Moreover, width inference occurs early in the compilation pipeline, undetected errors 
at this stage may propagate to subsequent passes and violate functional guarantees with high probability. Therefore,  
the problem of width inference warrants a systematic investigation in the presence of circular dependencies. For instance, it is unclear, \emph{a priori}, whether a unique least solution exists.}

\revise{Our long-term goal is to develop a formally verified FIRRTL compiler. This work makes a significant step toward that goal.
We formalize the width inference problem in {\firrtl}, investigate its theoretical properties,  propose a complete procedure for solving the problem, and develop a fully verified InferWidths pass.}

\subsubsection*{Contributions.}
The main contributions of this work are summarized as follows. 
\begin{itemize}
\item We formalize the width inference problem of the {\firrtl} language 
as the problem of solving \textbf{FIR}RTL \textbf{W}idth \textbf{INE}quality (FIRWINE)  constraints. In particular, we show that if a FIRWINE constraint is satisfiable, it has a unique least solution. 
\item We propose a complete procedure for deciding the satisfiability of FIRWINE constraints
and computing the least solution of satisfiable constraints. 
\item We implement our procedure in the interactive theorem prover Rocq\footnote{\url{https://rocq-prover.org}, previously known as Coq.} and prove its functional correctness.  
\item We extract an OCaml implementation from the Rocq implementation, as the first formally verified InferWidths pass. Extensive experiments
on a large number of benchmarks including three real-world RISC-V processor Chisel designs (i.e., Rocket Chip~\cite{RocketChip}, RISC-V BOOM \cite{Boom}, NutShell~\cite{NutShell}) show that our OCaml implementation can solve more width inference instances than the InferWidths pass in {\firtool}, with higher efficiency in general. 
\end{itemize}

We remark that the width inference problem can be reduced to 
integer linear programming (ILP) 
and solved by ILP solvers (e.g., Gurobi~\cite{Gurobi}). However, general-purpose ILP solving algorithms are too complicated to be formally verified currently, while our approach is bespoken and 
amenable to formal verification, which is a milestone towards 
a formally verified {\firrtl} compiler. 

\revise{
While we only target the problem from {\firrtl} in this work, 
analogous problems 
do exist in other hardware description languages. For instance, Verilog and VHDL allow developers to implement parameterized modules with configurable widths, thus the consistency of interconnects between such modules should be checked. 
Our approach could be adapted therein, where 
the consistency checking is reduced to satisfiability checking of the width constraints~\cite{SMTGO09}. In particular, the use of dependency graphs to resolve circular dependencies in the decision procedure is language-agnostic and therefore broadly applicable.}

\subsubsection*{Outline.} 
Section~\ref{sec-firrtl} introduces {\firrtl}, formalizes the width inference problem, and
shows that a unique least solution always exists 
for any satisfiable FIRWINE constraint.
Section~\ref{sec-procedure} presents our complete procedure. 
Section~\ref{sec-coq-proof} describes the Rocq implementation and its formal verification.
In Section~\ref{sec-impl-eval}, we extract the OCaml implementation and thoroughly evaluate its performance.
Finally, we discuss related work in Section~\ref{sec-related} and conclude in Section~\ref{sec-conclusion}.



\section{FIRRTL and the width inference problem}\label{sec-firrtl} 


%
A {\firrtl} \emph{circuit} consists of a list of modules, where a designated top-level module has the same identifier as the circuit.
Each module represents a hardware block that can be instantiated:
the top module is always instantiated, while other modules are instantiated as often as requested.
For instance, the {\firrtl}  circuit \prettyll{A} shown in Fig.~\ref{fig-infer-width-circular} has a (top-level) module \prettyll{A}.
(The formal syntax of the {\firrtl} language is given in \arxiv{Appendix~\ref{appendix:syntax}}\crv{\cite{ESOP26-full}}.) 

A module consists of a sequence of declarations of ports and 
statements. Each port has a direction (\prettyll{input} or \prettyll{output}), an identifier and a data type.
For instance, the module \prettyll{A} in Fig.~\ref{fig-infer-width-circular} has two input ports, i.e., \prettyll{in} and \prettyll{clock}, and one output port, i.e.,  \prettyll{out}, with types \prettyll{UInt<4>}, \prettyll{Clock} and  \prettyll{UInt}, respectively.
Each statement defines a wire, register, node, connect (\prettyll{<=}), conditional (\prettyll{when}), etc.
A type for specifying the structure of the data can be either a ground type, e.g., \prettyll{UInt<4>} (unsigned integer of bit width 4) or \prettyll{Clock} (carrying a clock signal);
or an aggregate type, which is either a vector type, e.g., \prettyll{UInt[3]} (a three-element vector of unsigned integers),
or a bundle type, e.g., \prettyll{\{enable: UInt<1>, calc: UInt<3>\}} comprising two signals of the ground types \prettyll{UInt<1>} and \prettyll{UInt<3>}, respectively.


A (circuit) component is a declared element in a {\firrtl}  circuit that can be referred to by its name; it can be a port, node, wire, register, memory, or module instance.
A component can appear in a connect statement.
This must respect the data flow, e.g., input ports and nodes can only be read from.

\subsection{The width inference problem in \firrtl}
{\firrtl} features high-level constructs such as vector types, bundle types, conditional statements, aggregate connects and modules, which are then gradually removed by a sequence of lowering transformations. 
A necessary step is to infer the minimum (bit) widths for all the components with unspecified widths while maintaining the legality of all incoming connections, as illustrated in Figs.~\ref{fig-infer-width-acyclic} and \ref{fig-infer-width-circular}.
An exception indicating that a certain width is uninferrable 
should be thrown if it fails to do so.
%
For instance, if {\tt x} is a wire with two connect statements: \prettyll{x <= UInt<5>(15)} and \prettyll{x <= UInt<1>(1)}, 
then the width of {\tt x} is inferred to be $5$,
where \prettyll{UInt<n>(m)} denotes the unsigned integer $m$ of bit width $n$. In contrast, if the wire {\tt x} is connected with the statement: \prettyll{x <= add(x, UInt<1>(1))}, then no width of {\tt x} exists, because the inequality $w_{{\tt x}} \ge \max(w_{\tt x}+1,1)$ is unsatisfiable. As a result, an exception 
would be thrown.


Throughout this paper, let $X= \{x_1, \ldots, x_n\}$ denote a nonempty set of variables ranging over the set  $\natnum$ of nonnegative integers. 
Moreover, for $m \ge 1$, let $[m]$ denote the set $\{1, \ldots, m\}$.
Formally, the \emph{width inference problem} for a given {\firrtl} program is to compute the \emph{least} (nonnegative) solution of the set of inequalities consisting of 
$w_{{\tt x}} \ge w_{\tt e}$ for each connect statement \prettyll{x <= e},
where $w_{{\tt x}}$ is a nonnegative integer variable denoting the unspecified width of $x$;
$w_{\tt e}$ denotes the width of the expression {\tt e} and is defined as in the {\firrtl} specification~\cite{firrtl-spec}
(\arxiv{cf.\ Appendix~\ref{sec:widthofExpr}}\crv{cf.\ \cite{ESOP26-full}} for a formal definition).
For example, $w_{\tt mux(a, b, c)} = \max(w_{\tt b}, w_{\tt c})$ and $w_{\tt add(a, b)} = \max(w_{\tt a}, w_{\tt b})+1$.  
By \cite[Section~25]{firrtl-spec}, $w_{\tt e}$ for each expression {\tt e} in {\firrtl} is a term following the grammar
\begin{center}
    $t \eqdef w_{\tt x} \mid c \mid \min(t, t) \mid \max(t, t) \mid t + t,$
\end{center} 
where $w_{\tt x}$ is a nonnegative integer variable denoting the unspecified width of a component {\tt x} and $c \in \intnum$. Note that the term $\min(w_{\tt x}, w_{\tt y})$ is introduced 
for the expression \prettyll{rem(x, y)} (i.e., the remainder when dividing {\tt x} by {\tt y}). Here we skip the \prettyll{dshl} operation in {\firrtl}, which would introduce terms of the form $2^{w_e}$. Our approach could be extended to deal with them, the detail of which, however, is left as future work. Although we do not support the most general dynamic shift left \prettyll{dshl}, all occurrences of \prettyll{dshl} in our benchmarks are in two specific forms that can be handled by our approach (cf. Section~\ref{sec-evaluation}).

The width inference problem is said to be \emph{satisfiable} if the conjunction of the inequalities of the form $w_{{\tt x}} \ge w_{\tt e}$ is satisfiable. A \emph{solution} of the  width inference problem is an assignment of the variables satisfying all the inequalities.

\subsection{Existence of a unique least solution}
As the solution space is not equipped with a natural total order, 
the width inference problem is meaningful only if \emph{a unique least solution exists} for a satisfiable set of inequalities.
We will show that this indeed is the case. 
(If multiple minimal solutions exist, the resulting hardware would be underspecified.)
%
%
%

The \emph{width terms} are defined by the rules 
\begin{center}
 $t \eqdef x_j \mid c \mid t + t$   
\end{center}
such that $x_j\in X$ and $c \in \intnum$. Note that in width terms, subtraction of variables 
is disallowed. 
As a result, we can assume that \emph{each width term is of the form $a_0 + \sum_{j \in [n]} a_j x_j$ such that $a_0 \in \intnum$ and $a_j \in \natnum$ for each $j \in [n]$}.

\begin{definition}[$\fwc$-constraints]
	A 
	\textbf{FIR}RTL \textbf{W}idth \textbf{INE}quality (FIRWINE) constraint ($\fwc$-constraint for short) over $X$ 
	is a conjunction of inequalities of the form $x_i \ge \min(t_{i,1}, \ldots, t_{i, k_i})$, 
	 each of which is referred to as an inequality of $x_i\in X$, where 
	$t_{i,1}, \ldots, t_{i, k_i}$ are width terms over $X$. 
\end{definition}

%
%
A solution of a $\fwc$-constraint $\varphi$ is a function $\eta: X \rightarrow \natnum$ such that each inequality in $\varphi$ holds under $\eta$. 
For brevity, we write a solution $\eta$ as $(\eta(x_1), \ldots, \eta(x_n))$. 
Let $\preceq_n$ denote the partial order over $\natnum^n$ such that $(u_1, \ldots, u_n) \preceq_n (v_1, \ldots, v_n)$ if for each $i \in [n]$, $u_i \le v_i$. 
$(u_1, \ldots, u_n)$ is said to be the \emph{least} solution of  a $\fwc$-constraint $\varphi$, 
if for each solution $(v_1, \ldots, v_n)$ of $\varphi$, $(u_1, \ldots, u_n) \preceq_n (v_1, \ldots, v_n)$.

Each inequality $w_{{\tt x}} \ge w_{\tt e}$ in a width inference problem can be rewritten as a conjunction of $\max$-free inequalities,
as $\min$ and addition distribute over $\max$ (and $\min$)
%
and $w_{{\tt x}} \ge \max(w_{\tt a}, w_{\tt b})$ is equivalent to 
$w_{{\tt x}} \ge w_{\tt a}\wedge w_{{\tt x}} \ge  w_{\tt b}$. Formally, 

\begin{proposition}\label{prop-rewrite}
	Each inequality $w_{{\tt x}} \ge w_{\tt e}$ formulated in a width inference problem can be transformed into an equivalent $\fwc$-constraint. 
\end{proposition}

\begin{example}
Consider the {\firlis}-constraint $\varphi \equiv x_1 \ge 2 x_2 - 3 \wedge x_2 \ge 2x_1 + 1.$
It is easy to see that $x_1 \ge 2x_2 -3 \ge 2(2x_1 +1)-3=4x_1 -1$. As a result, we can deduce that $1 \ge 3x_1$, and thus $x_1 = 0$. Moreover, from $0 \ge 2x_2 -3$ and $x_2 \ge 1$, we can get that $x_2= 1$. We conclude that $(0, 1)$ is the unique solution of $\varphi$, thus also the least solution of $\varphi$.

On the other hand, if $\varphi \equiv x_1 \ge 2 x_1 \wedge x_1 \ge 1$, then $0 \ge x_1 \ge 1$, meaning that $\varphi$ is unsatisfiable.\qed 
\end{example}

\begin{proposition}\label{prop-min-sol}
A satisfiable {\firlis}-constraint has a unique least solution.
\end{proposition}

\begin{proof}

Let $\varphi$ be a satisfiable $\fwc$-constraint. 
It suffices to show that:  
\begin{center}
  \begin{tabular}{|@{~~}l@{~~}|}
        \hline
        \rule{0pt}{3ex}If $\eta_1 = (u_1, \ldots, u_n)$ and $\eta_2 = (v_1, \ldots, v_n)$ are solutions of $\varphi$, \\
        \rule[-1.7ex]{0pt}{0pt}then $\eta = (\min(u_1, v_1), \ldots, \min(u_n, v_n))$ is also a solution of $\varphi$. \\
        \hline
  \end{tabular}
\end{center}

Below, we will write $\eta(x_i)$ to denote $\min(u_i,v_i)$ and similarly $\eta_1(x_i) = u_i$, $\eta_2(x_i) = v_i$;
and we will extend these functions to width terms.
%
%
%
Let $I$ denote the set of indices $i \in [n]$ such that $u_i \le v_i$. Then for each $i \in [n] \setminus I$, we have $u_i > v_i$. As a result,  for each $i\in [n]$, 
\begin{center}
  $\eta(x_i) = \min(u_i,v_i)= \begin{cases}
  u_i,  & \text{if}\ i\in I; \\
  v_i,  & \text{if}\ i\in [n] \setminus I.
\end{cases}$  
\end{center} 

For some $i \in [n]$, let $x_i \ge \min(t_{i,1}, \ldots, t_{i, k_i})$ be an inequality for the variable $x_i$ in the $\fwc$-constraint $\varphi$, 
where for each $j \in [k_i]$, $t_{i,j} = a_{i,j,0} + \sum_{l \in [n]} a_{i, j, l} x_l$.
We shall show that $\eta(x_i) \ge \min(\eta(t_{i,1}), \ldots, \eta(t_{i, k_i}))$.
Then it follows directly that $\eta$ is a solution of $\varphi$.



%
Let's first look at the case $u_i \le v_i$, i.e., $i\in I$.
Then
\begin{align*}
\eta(x_i) & = u_i \stackrel{*}{\ge} \min(\eta_1(t_{i,1}), \ldots, \eta_1(t_{i, k_i})) \\
& = \min\Big\{a_{i, j, 0} + \sum \limits_{l \in I} a_{i, j, l} \eta_1(x_l) + \sum\limits_{l \in [n] \setminus I} a_{i, j, l} \eta_1(x_l)\ \big \vert\ j \in [k_i] \Big\}\\
& = \min\Big\{a_{i, j, 0} + \sum \limits_{l \in I} a_{i, j, l} u_l + \sum\limits_{l \in [n] \setminus I} a_{i, j, l} u_l\ \big \vert\ j \in [k_i] \Big\}\\
& \stackrel{**}{\ge} \min\Big\{a_{i, j, 0} + \sum \limits_{l \in I} a_{i, j, l} u_l + \sum \limits_{l \in [n] \setminus I} a_{i, j, l} v_l\ \big \vert\ j \in [k_i] \Big\} \\
& = \min\Big\{a_{i, j, 0} + \sum \limits_{l \in I} a_{i, j, l} \eta(x_l) + \sum \limits_{l \in [n] \setminus I} a_{i, j, l} \eta(x_l)\ \big \vert\ j \in [k_i] \Big\} \\
& = \min(\eta(t_{i,1}), \ldots, \eta(t_{i, k_i})). 
\end{align*}
The first inequality $\stackrel{*}{\ge}$ above holds because $\eta_1$ is a solution of $\varphi$,
and the second inequality $\stackrel{**}{\ge}$ holds because $a_{i, j, l} \ge 0$ and for $l \in [n] \setminus I$, we have $u_l > v_l$.
%

The remaining case, namely $u_i > v_i$, is proven analogously, with the sets $I$ and $[n] \setminus I$ swapped.
\qed
\end{proof}


\section{A complete procedure for width inference}\label{sec-procedure}



Observing that $x_i \ge \min(t_1, \ldots, t_k)$ is equivalent to $\bigvee_{j \in [k]} x_i \ge t_j$, a {\firlis}-constraint can be rewritten into 
$\Phi_{\rm nf}:=\bigvee_{j_1 \in [\ell]} \bigwedge_{j_2 \in [m_{j_1}]} z_{j_1, j_2} \ge t_{j_1, j_2}$ such that $\Phi_{\rm nf}$
is satisfiable iff $\bigwedge_{j_2 \in [m_{j_1}]} z_{j_1, j_2} \ge t_{j_1, j_2}$ is satisfiable for some $j_1\in [\ell]$, where $z_{j_1, j_2} \in X$ and $t_{j_1, j_2}$ is a width term. 
%
%
Moreover, the least solution of $\Phi_{\rm nf}$ is the least one among the least solutions of the disjuncts.  
Hence, it is sufficient to  
consider a conjunction of inequalities of the form $x_i \ge a_0 + \sum_{j \in [n]} a_j x_j$. 
%

The procedure utilizes a concept of dependency graphs. 
\begin{definition}[Dependency graph]
Let $\varphi$ be a {\firlis}-constraint over $X$.
Assume that the number of inequalities for all variables $x_i$ together is at most $L\in \natnum$.
%
This allows to assign a unique label $l \in [L]$ to each inequality for $x_i$.

The dependency graph $G_\varphi$ of $\varphi$ is defined as a directed graph $(V_\varphi, E_\varphi)$, where $V_\varphi=X$ and $E_\varphi \subseteq V_\varphi \times [L] \times \natnum \times V_\varphi$ comprises the edges $(x_i, l, a_{j}, x_{j})$ such that $\varphi$ contains an inequality $l: x_i \ge a_0 + \sum_{j \in [n]} a_j x_j$ with $a_{j} > 0$. 
\end{definition}

\begin{example}\label{exmp-dep-graph}
Given the labeled {\firlis}-constraints $\varphi_1$, $\varphi_2$ and $\varphi_3$:
\vspace{-1.5mm}
\begin{align*}
\varphi_1 \equiv {} & \begin{array}[t]{l@{{}\quad\wedge\quad{}}l@{{}\quad\wedge\quad{}}l}
                    l_1: x_1 \ge 5 & l_2: x_2 \ge 2 & l_3: x_3 \ge 2x_1 - 3 \quad\wedge {} \\
                    l_4: x_3 \ge x_2 +2 & l_5: x_4 \ge 3 x_2 & l_6: x_4 \ge x_3 + 4, \end{array} \\ \specialrule{0em}{2pt}{2pt}
\varphi_2 \equiv {} & l_1: x_1 \ge 2x_2 - 4 \quad\wedge\quad l_2: x_2 \ge x_3 -2 \quad\wedge\quad l_3: x_3 \ge x_1 + 1, \\ \specialrule{0em}{2pt}{2pt}
\varphi_3 \equiv {} & \begin{array}[t]{l@{{}\quad\wedge\quad{}}l@{{}\quad\wedge\quad{}}l@{}l}
                    l_1: x_1 \ge 2x_2 - 4 & l_2: x_2 \ge x_1 + 1 & l_3: x_4 \ge x_2 & {}\quad\wedge\quad {} \\
                    l_4: x_3 \ge x_4\ & l_5: x_5 \ge x_3 + 1 & l_6: x_5 \ge 2 & {}\quad\wedge\quad {} \\
                    l_7: x_6 \ge x_5 - 1 & l_8: x_3 \ge x_6 & \multicolumn{2}{l}{l_9: x_7 \ge x_3.} \end{array}
\end{align*}
\vspace{-1.5mm}

\begin{figure}[t]
\centering
\includegraphics[width=0.99\textwidth]{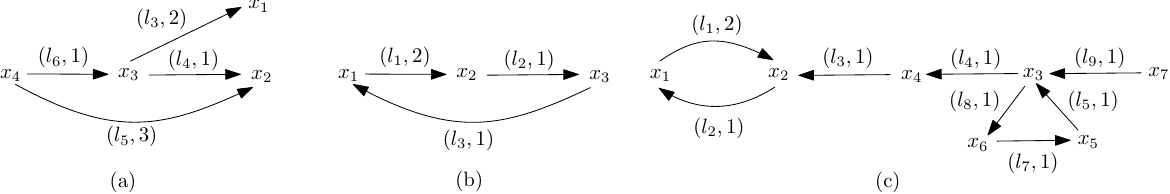}
\vspace{-2mm}
\caption{The dependency graphs of $\varphi_1, \varphi_2, \varphi_3$}\label{fig-dep-graph}
\vspace{-3mm}
\end{figure} 
The dependency graphs for these {\firlis}-constraints are $G_{\varphi_1}$, $G_{\varphi_2}$, and $G_{\varphi_3}$, as illustrated in Fig.~\ref{fig-dep-graph}(a)-(c).
It is easy to see that 
$G_{\varphi_1}$ is acyclic, 
$G_{\varphi_2}$ is strongly connected,
and $G_{\varphi_3}$ comprises four strongly connected components, two of which are nontrivial (containing at least one edge).\qed
\end{example}

\subsection{The overall procedure} 
Let $\sccgraph_{G_\varphi}$ be the SCC graph of $G_\varphi$. We first compute a topological sort $C_1, \ldots, C_k$ of $\sccgraph_{G_\varphi}$, namely, any edge from the vertices in $C_i$ to the vertices in $C_{j}$ entails $i \leq j$.
We denote by $V(C)$ the set of variables in $C$.
For example, in Fig.~\ref{fig-dep-graph}(c), we have $V(C_1) = \{ x_7 \}$, $V(C_2) = \{ x_3, x_5, x_6\}$, $V(C_3) = \{ x_4 \}$ and $V(C_4) = \{ x_1, x_2\}$.
We then iterate the following  procedure from $j= k$ to $j=1$:
\begin{center}
\fbox{\begin{minipage}{\dimexpr\textwidth-0.5cm}
Suppose that the least width $u_{i} \in \natnum$ of each $x_{i} \in V(C_{j+1}) \cup \cdots \cup V(C_k)$ has been computed. 
We first replace each occurrence of $x_{i} \in V(C_{j+1}) \cup \cdots \cup V(C_k)$ in the inequalities for $x_{i'}\in V(C_j)$ with the least width $u_{i}$, 
resulting in a conjunction $\psi_{C_j}$ of inequalities 
for $x_{i'} \in V(C_j)$.
\begin{itemize}
    \item If $C_j$ is trivial, i.e., it comprises $x_{i'}$ only, then we set $u_{i'}$ to be the maximum of $0$ and the constants $a$ in the inequalities $x_{i'} \ge a$ of $\psi_{C_j}$.
    \item Otherwise $C_j$ is strongly connected, we compute the least solution $(u_{i'})_{x_{i'} \in C_j}$ for $\psi_{C_j}$, or $\psi_{C_j}$ is unsatisfiable and the procedure terminates.  
\end{itemize}
\end{minipage}
}
\end{center} 

We shall now focus on the computation of the least solution for \emph{strongly connected dependency graphs,} which is the core of our procedure.
We first 
identify a class of strongly connected dependency graphs via expansiveness.

\subsection{Characterization of strongly connected dependency graphs}\label{sec-scc}
The notion of \emph{expansiveness} is to provide a structural characterization for a class of dependency graphs.
A dependency graph $G_\varphi$ is said to be \emph{expansive} if it contains an edge $(x, l, a, x')$ with $a > 1$ or two distinct edges out of some vertex with the same label; otherwise it is 
\emph{nonexpansive}.  
Expansive dependency graphs exhibit the following property.

\begin{proposition}\label{prop-g-1}
Assume $G_\varphi$ is strongly connected and $\varphi$ is satisfiable. $G_\varphi$ is expansive iff 
there is a constant $B \in \natnum$ such that  the values of  the variables in every solution of $\varphi$ are upper-bounded by $B$. 
\end{proposition}

 

\begin{proof}
Note that any edge $(z_{r-1}, l_r, a_r, z_r)$ in $G_\varphi$ stems from an inequality $l_r : z_{r-1} \geq a_{r,0} + \sum_{z \in X} a_{r,z} z$, where $a_{r,z_r} = a_r > 0$. We can weaken this inequality to $z_{r-1} \geq a_{r,0} + a_r z_r$, and we denote this weakened inequality by $\mathsf{Ineq}(z_{r-1}, l_r, a_r, z_r)$.
Similar weakened inequalities can be defined over paths inductively as follows.
Given the empty path $\pi_{z_0 \to z_0}$ from $z_0$ to $z_0$, let $$\mathsf{Ineq}(\pi_{z_0 \to z_0}) \equiv z_0 \geq z_0.$$ 

For a non-empty path $\pi_{z_0 \to z_m} = \pi'_{z_0 \to z_{m-1}} \cdot (z_{m-1}, l_m, a_m, z_m)$ from $z_0$ via $z_{m-1}$ to $z_m$ and the inequalities
\begin{align*}
\mathsf{Ineq}(\pi'_{z_0 \to z_{m-1}}) \equiv \phantom{z_{m-1}}\makebox[0pt][r]{$z_0$} & \geq a_{m-1} + b_{m-1} z_{m-1} \\
\mathsf{Ineq}(z_{m-1}, l_m, a_m, z_m) \equiv                                  z_{m-1} & \geq a_{m,0} + a_m z_m, \\
\intertext{let}
               \mathsf{Ineq}(\pi_{z_0 \to z_m}) \equiv \phantom{z_{m-1}}\makebox[0pt][r]{$z_0$} & \geq a_{m-1} + b_{m-1} (a_{m,0} + a_m z_m) \\
                                                                                                 & = a_{m-1} + b_{m-1}a_{m, 0} + b_{m-1}a_m z_m
.
\end{align*}
Any solution of $\varphi$ satisfies all such $\mathsf{Ineq}(\pi)$, for all paths $\pi$ in $G_\varphi$.

\begin{description}
\item[-- If] \textbf{$G_\varphi$ is expansive, then the solutions are upper-bounded:}
Assume that $G_\varphi$ is expansive.
We prove that there exists an upper bound for an arbitrary variable $x_k$.
There are two cases:
\begin{description}
\item[$G_\varphi$] \textbf{contains an edge $(x_i, l, a, x_j)$ with $a > 1$:}
	Because $G_\varphi$ is strongly connected, there exist paths $\pi_{ki}$ from $x_k$ to $x_i$
	and $\pi_{jk}$ from $x_j$ to $x_k$.
	The inequality for the composed path $\mathsf{Ineq}(\pi_{ki} \cdot (x_i, l, a, x_j) \cdot \pi_{jk})$ is of the form $x_k \geq a_{k,0} + a_{kk} x_k$.
	Note that we must have that $a_{kk} > 1$ because it is a multiple of $a > 1$.
	This means that any solution of $\varphi$ satisfies $x_k \leq -a_{k,0} / (a_{kk}-1)$.
\item[$G_\varphi$] \textbf{contains two distinct edges $(x_i, l, 1, x_{j_1})$ and $(x_i, l, 1, x_{j_2})$:}
	We can assume that all edges $(x', l', a', x'')$ in $G_\varphi$ have $a'=1$;
	otherwise the result immediately follows from the first case.
	
	From these two edges, we can derive a similar weakening of inequality $l$,
	namely $x_i \geq a_{i,0} + x_{j_1} + x_{j_2}$.
	There exist the following paths with their inequalities:
	\vspace{-1.5mm}
	\begin{align*}
		& \makebox[0pt][l]{$\pi_{ki}$}\phantom{\pi_{j_1k}} \text{ from $x_k$ to $x_i$}       & \mathsf{Ineq}(\pi_{ki}) \equiv \phantom{x_{j_1}}\makebox[0pt][r]{$x_k$} & \geq a_{ki} + x_i \\
		& \pi_{j_1k} \text{ from $x_{j_1}$ to $x_k$} & \mathsf{Ineq}(\pi_{j_1k}) \equiv x_{j_1} & \geq a_{j_1k} + x_k \\
		& \pi_{j_2k} \text{ from $x_{j_2}$ to $x_k$} & \mathsf{Ineq}(\pi_{j_2k}) \equiv x_{j_2} & \geq a_{j_2k} + x_k.
	\end{align*}
	\vspace{-1.5mm}
	Combining all these, we get that
	$$x_k \geq a_{ki} + x_i \geq a_{ki} + a_{i,0} + x_{j_1} + x_{j_2} \geq \linebreak[2] a_{ki} + a_{i,0} + a_{j_1k} + x_k + a_{j_2k} + x_k.$$
	This means that any solution of $\varphi$ satisfies $x_k \leq -a_{ki} - a_{i,0} - a_{j_1k} - a_{j_2k}$.
\end{description}
\item[-- If] \textbf{$G_\varphi$ is nonexpansive, then the solutions are not upper-bounded:}
Assume that
$G_\varphi$ does not contain two distinct edges out of some vertex $x_i$ with the same label, 
nor contains any edge $(x_i, l, a, x_j)$ such that $a > 1$. 
Therefore all inequalities in $\varphi$ are of the form $x_i \geq a'_{ij} + x_j$ or $x_i \geq a''_i$ for some constants $a'_{ij}, a''_i \in\intnum$.

Since $\varphi$ is satisfiable, let
$(u_1, u_2, \ldots, u_n)$ be a solution of $\varphi$.
Obviously, for any $C \in \natnum$,  $(u_1 + C, u_2 + C, \ldots, u_n + C)$ is also a solution of $\varphi$.
Thus, the values of the variables in the solutions of $\varphi$ are \emph{not} upper-bounded.
\qed
\end{description}
\end{proof}

As a result, if $G_\varphi$ is expansive, there are finitely many solutions for $\varphi$ and the least solution can be computed by conducting an exhaustive search. 
Our procedure therefore diverges into separate cases depending on whether $G_\varphi$ is expansive (Section~\ref{sec-expansive}) or nonexpansive (Section~\ref{sec-nonexpansive}).


\subsection{Computing the least solution for expansive $G_\varphi$}\label{sec-expansive}

In this case, we first present an algorithm to compute the upper bounds (which follows essentially the proof of  Proposition~\ref{prop-g-1}), then we show how to utilize them to search for the least solution based on a branch-and-bound algorithm. 


\subsubsection*{Computing upper bounds for the variables.}
The algorithm distinguishes the two alternative cases for the notion of expansiveness.
\begin{itemize}
\item If there is an edge $(x_i, l, a, x_j)$ with $a > 1$, then we compute the upper bounds  by the following procedure. 
\begin{enumerate}
\item Find a path from $x_j$ to $x_i$, say, $\pi_{ji} = (y_0, l_1, a_1, y_1) \cdots (y_{r-1}, l_r, a_r, y_r)$ (where $y_0 = x_j$ and $y_r = x_i$).

\item Calculate $\mathsf{Ineq}((x_i, l, a, x_j) \concat \pi_{ji})$ as in the proof of Proposition~\ref{prop-g-1}, say, $x_i \geq a_{i,0} + a_{ii} x_i$.
%
\item From $a > 1$, we know that $a_{ii} > 1$.
As a result, $x_i \le \floor{-a_{i,0}/(a_{ii}-1)}$. We have obtained an upper bound for $x_i$.
\item We start from $x_i$ and apply a breath-first search (BFS) to compute upper bounds for the other variables. That is, suppose that we know $x_{i'} \leq c$ and $(x_{i'}, l', a', x_{j'})$ is an edge in $G_\varphi$, let $l': x_{i'} \ge b_0 + \sum_{k \in [n]} b_k x_k$ (where $b_{j'} = a'$), then we have $a' x_{j'} \le x_{i'} - b_0 \leq c - b_0$, thus $x_{j'} \le \floor{(c - b_0)/a'}$. (Because $G_\varphi$ is strongly connected, applying the BFS produces upper bounds for all the variables in $G_\varphi$.) 
\end{enumerate}

\item If there are two distinct edges out of some vertex $x_i$ of the same label, say, $e_1 = (x_i, l, a_1, x_{j_1})$ and $e_2 = (x_i, l, a_2, x_{j_2})$, then we compute the upper bounds by the following procedure. Let us assume that the weights of all the edges in $G_\varphi$ are equal to $1$. Because otherwise, we already know how to compute the upper bounds by the aforementioned procedure. 
\begin{enumerate}
    \item Find two simple paths from $x_{j_1}$ and $x_{j_2}$ to $x_i$, respectively, say $\pi_1$ and $\pi_2$. Then as in the proof of Proposition~\ref{prop-g-1}, we can utilize $\mathsf{Ineq}(e_1 \concat \pi_1)$ and $\mathsf{Ineq}(e_2 \concat \pi_2)$ to compute an upper bound for $x_i$, say $x_i \leq B_i$. 
    \item Similarly to the situation that there is an edge of weight greater than $1$ in $G_\varphi$, we start from $x_i$ and apply a BFS to compute upper bounds for all the other variables in $G_\varphi$.  
\end{enumerate}
\end{itemize}


\begin{example}
Consider $\varphi_2$ in Example~\ref{exmp-dep-graph}. The dependency graph $G_{\varphi_2}$ as illustrated in Fig.~\ref{fig-dep-graph}(b) is strongly connected and contains an edge $(x_1, l_1, 2, x_2)$ whose weight is greater than $1$. Then we compute the upper bounds as follows. 

Because $x_1 \ge 2x_2 - 4$, we find a simple path from $x_2$ to $x_1$, say $(x_2, l_2, 1, x_3) \concat (x_3, l_3, 1, x_1)$, and utilize the path to apply the replacements. Then we have 
\begin{center}
  $x_1 \ge 2x_2 -4 \ge 2(x_3-2)-4 = 2x_3 -8 \ge 2(x_1+1)-8 = 2x_1 -6.$
\end{center} 
 As a result, $x_1 \le 6$. 

Finally, we start from $x_1$ and apply a BFS to compute the upper bounds for $x_2$ and $x_3$. At first, from $x_1 \ge 2x_2 -4$, we have $2x_2 \le x_1+4 \le 6+4=10$, thus $x_2 \le 5$. Moreover, from $x_2 \ge x_3-2$, we have $x_3 \le x_2+2 \le 7$. 

To summarize, we obtain the upper bounds $(6,5,7)$ for $(x_1, x_2, x_3)$.
\end{example}


\subsubsection*{Branch-and-bound.}
Let $(ub_1, \ldots, ub_n)$ denote the upper bounds of the variables in $X$. 
For each $i \in [n]$, let $lb_i$ denote the maximum of $0$ and the constants $a_0$ in the inequalities $x_i \ge a_0 + \sum_{j \in [n]} a_j x_j$ of $\varphi$. If $lb_i > ub_i$ for some $i\in [n]$, then report that $\varphi$ is unsatisfiable and terminate the algorithm. Otherwise, the algorithm iterates the following branch-and-bound procedure, where the upper bounds $(ub_1, \ldots, ub_n)$ and lower bounds $(lb_1, \ldots, lb_n)$ are taken as parameters. (See Figure~\ref{fig:bab} for an illustration of the procedure.)  

\begin{figure}[t]
    \centering
    \includegraphics[scale = 0.78]{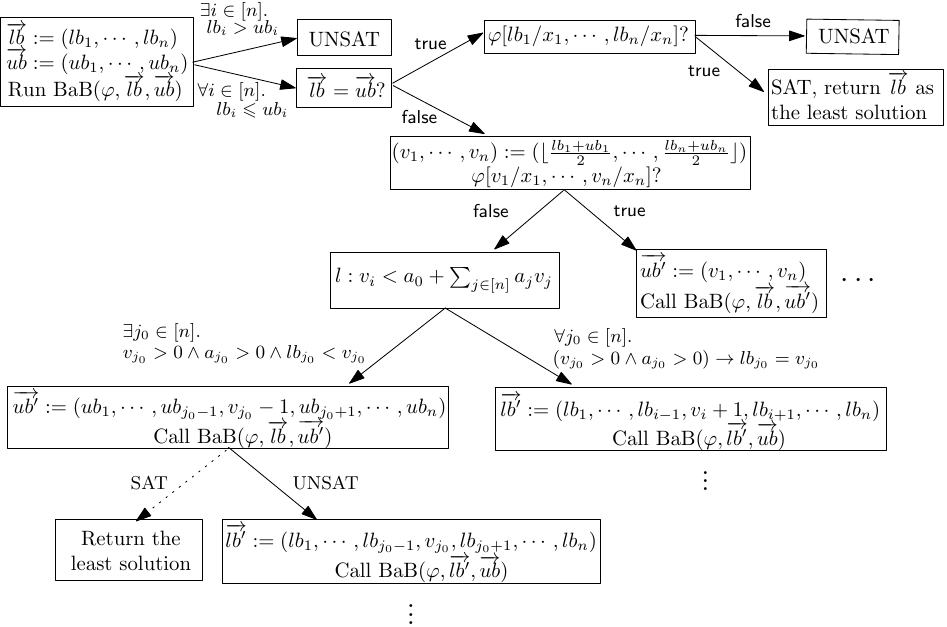}
    \caption{An illustration of the branch-and-bound (BaB) procedure}
    \label{fig:bab}
\end{figure}

  In the case $(ub_1, \ldots, ub_n) = (lb_1, \ldots, lb_n)$, we check whether
  $(lb_1, \ldots, lb_n)$ is a solution of
   $\varphi$ (i.e., $\varphi[lb_1/x_1, \ldots, lb_n/x_n]$ is $\ltrue$ or not). If it is $\ltrue$, then return $(lb_1, \ldots, lb_n)$ as the least solution; otherwise, report that $\varphi$ is unsatisfiable. 

 Otherwise $(ub_1, \ldots, ub_n) \neq (lb_1, \ldots, lb_n)$, then we compute $(v_1,\ldots, v_n) = (\floor{(ub_1+lb_1)/2}, \ldots, \floor{(ub_n+lb_n)/2})$ and 
 check whether $(v_1,\ldots, v_n) $ is a solution of $\varphi$. If it is a solution, then we set $(ub'_1, \ldots, ub'_n) := (v_1, \ldots, v_n)$ and attempt to compute the least solution of $\varphi$ w.r.t. $(ub'_1, \ldots, ub'_n) $ and $(lb_1, \ldots, lb_n)$.
%
%
 Otherwise, $(v_1,\ldots, v_n) $ violates some inequality $l : x_i \ge a_0 + \sum_{j \in [n]} a_j x_j$, i.e., $v_i < a_0 + \sum_{j \in [n]} a_j v_j$.  In this case, there must be $j_0 \in [n]$ such that $v_{j_0} > 0$ and $a_{j_0} > 0$, because otherwise $v_i < a_0$, contradicting $v_i \ge lb_i \ge a_0$. We proceed by a case distinction. 
%
\begin{itemize}
 \item \underline{There exists $j_0 \in [n]$ such that $v_{j_0} > 0$, $a_{j_0} > 0$ and 
$lb_{j_0} < v_{j_0}$.} In this case, we try to satisfy the inequality by decreasing the right-hand side $a_0 + \sum_{j \in [n]} a_j x_j$,
concretely by decreasing $x_{j_0}$. To this end, we set $ub'_{j_0} := v_{j_0}-1$ and attempt to compute the least solution of $\varphi$ w.r.t. $(ub_1, \ldots, \linebreak[0] ub_{j_0-1}, ub'_{j_0}, \linebreak[0] ub_{j_0+1}, \linebreak[0] \ldots, \linebreak[0] ub_n)$ and $(lb_1, \ldots, lb_n)$.
\begin{itemize}
\item If the attempt succeeds,
we have computed the least solution of $\varphi$ w.r.t. the same lower bound, which is returned as the desired least solution.
\item If the attempt fails, there might still be a solution with $x_{j_0} \geq v_{j_0}$.
So let $lb'_{j_0} := v_{j_0}$ and return the least solution of $\varphi$ w.r.t. $(ub_1, \ldots, ub_n)$ and $(lb_1, \ldots, lb_{j_0-1}, lb'_{j_0}, lb_{j_0+1}, \ldots, lb_n)$. In the case that $\varphi$ is unsatisfiable w.r.t. these bounds, report that $\varphi$ is unsatisfiable.
\end{itemize}
\item \underline{For every $j_0 \in [n]$ such that $v_{j_0} > 0$ and $a_{j_0} > 0$, $lb_{j_0} = v_{j_0}$.} 
In this case, it is impossible to decrease the right-hand side of $l : x_i \ge a_0 + \sum_{j \in [n]} a_j x_j$,
so the only way to satisfy the inequality is to increase $x_i$.
Hence we set $lb'_{i}:=v_i+1$ and compute the least solution of $\varphi$ w.r.t. $(ub_1, \ldots, ub_n)$ and $(lb_1, \ldots, lb_{i-1}, lb'_{i}, lb_{i+1}, \ldots, lb_n)$. If $\varphi$ is unsatisfiable, 
report that $\varphi$ is unsatisfiable; otherwise, the least solution 
will be returned as the desired least solution.  

%
\end{itemize}

\begin{example}
Consider the {\firlis}-constraint 
\begin{center}
    $\varphi \equiv x_1 \ge 2x_2 - 4 \quad\wedge\quad x_2 \ge x_3 -2 \quad\wedge\quad x_3 \ge x_1 + 1.$
\end{center} 
Our procedure 
produces the upper bounds of the variables $(ub_1,ub_2,ub_3):=(6,6,8)$. The lower bounds are extracted from $\varphi$, that are $(lb_1,lb_2,lb_3):=(0,0,1)$, referring to the constants in the inequalities. Then, our procedure computes the least solution using 
the upper bounds $(ub_1,ub_2,ub_3)$ and lower bounds $(lb_1,lb_2,\linebreak[0] lb_3)$ as follows.
\begin{enumerate}
    \item  First, since  $lb_i\le ub_i$ for $1\leq i\leq 3$ but $(lb_1,lb_2,lb_3)\neq (ub_1,ub_2,ub_3)$,
    we compute  $(v_1,v_2,v_3)$ from the current upper bounds and lower bounds: 
    \[(v_1,v_2,v_3):=(\lfloor(lb_1+ub_1)/2\rfloor,
    \lfloor(lb_2+ub_2)/2\rfloor,\lfloor(lb_3+ub_3)/2\rfloor)
=(3,3,4).\] 
Since $3 \ge 2\times3 -4 \wedge 3\ge 4-2 \wedge 4\ge 3+1$, $(v_1,v_2,v_3)$ is a solution of $\varphi$. Thus, we set
$(ub_1',ub_2',ub_3') := (v_1,v_2,v_3)=(3,3,4)$ and compute
the least solution using the updated upper bounds $(ub_1',ub_2',ub_3')$ and previous lower bounds $(lb_1,lb_2,lb_3)$.
\item With the updated upper bounds $(ub_1',ub_2',ub_3')$, we have  $lb_i\le ub_i'$ for $1\leq i\le 3$, but $(lb_1,lb_2,lb_3)\neq (ub_1',ub_2',ub_3')$.
Thus, we compute the new assignment
$(v_1',v_2',v_3')$ from the updated upper bounds and previous lower bounds: \[(v_1',v_2',v_3'):=(\lfloor(lb_1+ub_1')/2\rfloor,\lfloor(lb_2+ub_2')/2\rfloor,\lfloor(lb_3+ub_3')/2\rfloor) 
=(1,1,2),\] 
which also satisfies $\varphi$ as $1 \ge 2\times 1-4\wedge1\ge2-2\wedge2\ge 1+1$.
Then, the upper bounds are updated to 
 $(ub_1'',ub_2'',ub_3''): =(v_1',v_2',v_3')= (1,1,2)$,
 with which the procedure is repeated to compute the least solution.
\item Similarly, since $lb_i\le ub_i''$ for $1\leq i\leq 3$ but $(lb_1,lb_2,lb_3)\neq (ub_1'',ub_2'',ub_3'')$, 
we compute the assignment
\[(v_1'',v_2'',v_3''):=(\lfloor(lb_1+ub_1'')/2\rfloor,\lfloor(lb_2+ub_2'')/2\rfloor,\lfloor(lb_3+ub_3'')/2\rfloor) 
=(0,0,1).\]
Since $(v_1'',v_2'',v_3'')$ is a solution of $\varphi$, the upper bounds are updated again, i.e., $(ub'''_1,ub'''_2,ub'''_3) :=(v_1'',v_2'',v_3'')= (0,0,1)$ with which the procedure is repeated.
\item At this moment, we have that $(lb_1,lb_2,lb_3)=(ub'''_1,ub'''_2,ub'''_3)=(0,0,1)$. Moreover, $(0,0,1)$ is a solution of $\varphi$. Thus, the procedure  returns $(0,0,1)$ as the least solution and terminates.\qed 
\end{enumerate}
\end{example}
 
\subsection{Computing the least solution for nonexpansive $G_\varphi$}\label{sec-nonexpansive}

%
In this case, every edge $(x_i, l, a, x_j)$ in $G_\varphi$ satisfies $a = 1$, and all the edges out of the same vertex have mutually distinct labels. Then each inequality in $\varphi$ is of the form $x_i \ge c_{ij} + x_j$ or $x_i \ge c'_i$ for some $c_{ij}, c'_i \in \intnum$.

We construct another graph $G'_\varphi$ comprising the edges $(x_i, c_{ij}, x_j)$ for the inequalities $x_i \ge c_{ij} + x_j$.
The computation of the least solution of $\varphi$ is reduced to determining the maximum distances (i.e., maximum weights of the paths)  between vertices in $G'_\varphi$. 

The intuition is as follows. For every $i \in [n]$, if $x_i \ge c_{ij} + x_j$ and $x_j \ge c_{jk} + x_k$, since $x_i,x_j,x_k$ range over $\natnum$, we can deduce that 
    $x_i \ge c_{ij}+x_j \ge c_{ij}$ and $x_i \ge c_{ij}+x_j \ge c_{ij}+c_{jk}+x_k \ge c_{ij}+c_{jk}$.
In general, for each path starting from $x_i$, say, $(y_1, l_1, a_1, y_2) \ldots (y_r, l_r, a_r, y_{r+1})$ (where $y_1 = x_i$), we have $x_i \ge \max(a_1, a_1+a_2, \ldots, \sum_{j \in [r]} a_j)$. 
More specifically, the least solution is computed as follows.
For each $i \in [n]$, we define $v_i$ as the maximum of $0$, and the constants $c'_i$ in the inequalities $x_i \ge c'_i$.
Then the value of $x_i$ in the least solution is the maximum of $v_i$ and the numbers $w_{ij}+v_j$ where $j \in [n]$ and $w_{ij}$ is the maximum distance
from $x_i$ to $x_j$ in $G'_\varphi$.

To compute the maximum distances between vertices in $G'_\varphi$, we adapt the well-known Floyd--Warshall algorithm~\cite{Floyd62,Warshall62} for computing the minimum distances between vertices into an algorithm for computing the maximum distances. Moreover, during the computation, if a cycle of positive weight is discovered, then report that $\varphi$ is unsatisfiable. Let us call the adapted algorithm the \emph{maximal Floyd--Warshall algorithm}. 

\begin{example}
Consider 
\begin{center}
    $\varphi \equiv l_1: x_1 \ge x_2 + 1 ~~\wedge~~ l_2: x_1 \ge 2 ~~\wedge~~ l_3: x_3 \ge x_1 - 1 ~~\wedge~~ l_4: x_2 \ge x_3.$
\end{center}  
Then $G_\varphi$ is a cycle comprising three edges, namely, $(x_1, l_1, 1, x_2)$, $(x_2, l_4, 1, x_3)$, and $(x_3, l_3, 1, x_1)$. Evidently, $G_\varphi$ is non-expansive. To solve the width inference problem, we construct $G'_\varphi$, which comprises the edges $(x_1, 1, x_2)$, $(x_2, 0, x_3)$, and $(x_3, -1, x_1)$. Then we use the maximal Floyd--Warshall algorithm to compute the following matrix $W$ denoting the maximum distances between vertices,  
\[\arraycolsep=4pt
W = 
\left[
\begin{array} {c c c}
0 & 1 & 1 \\
-1 & 0 & 0 \\
-1 & 0 & 0
\end{array}
\right].
\]
Since $w_{11} = w_{22} = w_{33} = 0$, i.e.,
no cycles have positive weights,
we get that $\varphi$ is satisfiable. 
Because $x_1 \ge 2$,  we define $(v_1, v_2, v_3)=(2, 0, 0)$.
As a result, the least solution is $(u_1, u_2, u_3) = (2, 1, 1)$, because 
\begin{itemize}
    \item $u_1= \max\{v_1,\ w_{12}+v_2,\ w_{13}+v_3\} = \max\{2,\ \phantom{-}1+0, 1+0\} = 2$,
    \item $u_2 = \max\{v_2,\ w_{21}+v_1,\ w_{23}+v_3\} = \max\{0,\ -1+2, 0+0\} = 1$, 
    \item  $u_3 = \max\{v_3,\ w_{31}+v_1,\ w_{32}+v_2\} = \max\{0,\ -1+2, 0+0\}= 1$. \qed
\end{itemize}
\end{example}

\subsection{A complete illustrative example}\label{sec-finalexample}
We present a complete example to demonstrate the entire procedure for solving the width inference problem. 
\begin{example}
Consider $\varphi_3$ in Example~\ref{exmp-dep-graph}. Its dependency graph $G_{\varphi_3}$, as shown in Fig.~\ref{fig-dep-graph}(c), has four SCCs, 
namely $V(C_1) = \{ x_7 \}$, $V(C_2) = \{ x_3, x_5, x_6\}$, $V(C_3) = \{ x_4 \}$ and $V(C_4) = \{ x_1, x_2\}$.
It is easy to see that $C_1 :: C_2 :: C_3 :: C_4$ is the unique topological sort of the SCC graph of $G_{\varphi_3}$.
Then we compute the least solution as follows.
\begin{enumerate}
\item At first, since $C_4$ (with vertices $x_1, x_2$) contains an edge $(x_1, l_1, 2, x_2)$, we compute the upper bounds for $x_1$ and $x_2$, which are $(2, 3)$. The lower bounds for $(x_1, x_2)$ are $(0,1)$.
Next, we apply the branch-and-bound algorithm to compute the least solution.
Namely, we consider $((2+0)/2, (3+1)/2) = (1, 2)$. Since $(1,2)$ satisfies $x_1 \ge 2x_2 - 4$ and $x_2 \ge x_1 + 1$, we strengthen the upper bounds to $(1,2)$.
Then we compute the least solution of $\varphi$ with respect to the upper bounds $(1,2)$ and lower bounds $(0,1)$.
Since $(0,1)$ satisfies both $x_1 \ge 2x_2 - 4$ and $x_2 \ge x_1 + 1$, we strengthen the upper bounds to $(0,1)$.
Then we return $(0,1)$ as the least solution since the lower and upper bounds coincide.
Therefore, $(u_1, u_2) = (0,1)$.
\item For $C_3$ (with vertex $x_4$), we replace $x_2$ in $x_4 \ge x_2$ with $u_2 = 1$ and obtain $x_4 \ge 1$. Then we set $u_4 = 1$.
\item For $C_2$ (with vertices $x_3, x_5, x_6$), we replace $x_4$ in $x_3 \ge x_4$ with $1$, and obtain $x_3 \ge 1$. Moreover, we have the following inequalities of $x_3, x_5, x_6$: $x_5 \ge x_3 + 1$, $x_5 \ge 2$, $x_6 \ge x_5 - 1$ and $x_3 \ge x_6$. Because $C_3$ is nonexpansive, we construct the graph $G'_{C_3}$ that comprises the edges $(x_5, 1, x_3)$, $(x_6, -1, x_5)$, and $(x_3, 0, x_6)$.
Then we utilize the maximal Floyd--Warshall algorithm to compute the maximum distances between vertices as the following matrix (where the vertex order is $x_3, x_5, x_6$), 
\[\arraycolsep=4pt
W = 
\left[
\begin{array} {c c c}
0 & -1 & 0 \\
1 & 0 & 1 \\
0 & -1 & 0
\end{array}
\right].
\]
Moreover, we compute $(v_3, v_5, v_6) = (1, 2, 0)$. As a result, $(u_3, u_5, u_6) = (\max(1,-1+1,0+1), \max(2, 1+1, 1+0), \max(0,0+1,-1+2))=(1,2,1)$.
\item Finally, for $C_1$ (with vertex $x_7$), we replace $x_3$ in $x_7 \ge x_3$ with $u_3$ and obtain $x_7 \ge 1$. Therefore, $u_7 = 1$.
\end{enumerate}
In the end, we obtain the least solution $(0,1,1,1,2,1,1)$.\qed
\end{example}

\subsection{Time complexity}
\revise{The overall worst-case time complexity of our procedure is exponential in the size of the dependency graph, i.e., exponential in the number of cicuit components (dominated by branch-and-bound for expansive SCCs). The time complexity of each stage is as follows.
\begin{itemize}
    \item Topological sorting. We use Tarjan’s algorithm, which runs in linear time in the size of the dependency graph, i.e., $O(|V| + |E|)$.
    \item Non-expansive SCCs. For each non-expansive SCC, we apply a maximal Floyd–Warshall procedure with time complexity $O(n^3)$, where $n$ is the number of nodes in the SCC.
    \item Expansive SCCs. For each expansive SCC, we apply a branch-and-bound (BaB) procedure. First, the upper bound $B$ in Proposition~\ref{prop-g-1} is computed in polynomial time in the size of the SCC. The solution is then searched over a space of size $B^n$ via branch-and-bound, which is exponential in $n$ in the worst case (e.g., when  all branches must be explored).
\end{itemize}
However, in practice, most SCCs encountered in our experiments are either small or non-expansive. Consequently, our procedure performs well in practice.}

\section{Rocq implementation and formal verification}\label{sec-coq-proof}

In this section, we formalize in Rocq the width inference procedure in Section~\ref{sec-procedure}, specify the functional correctness, and illustrate how to prove that the Rocq implementation satisfies the correctness. In total, the formalization has 5,301 lines of Rocq code; the specification and proof contain 
6,992 lines of Rocq code.


\subsection{Formalization of the width inference procedure}
Let $\varphi$ be a $\fwc$-constraint.
We represent $\varphi$ as a map $cm$ in Rocq, where $cm(x)$ for each variable $x$ is a set of inequalities for $x$. 
%
We introduce a data structure to represent the dependency graph $G_\varphi$ on which the
formally proved tarjan module in Mathcomp 2.0~\cite{tarjan} is applied to compute a topological sort of the SCC graph of $G_\varphi$, represented as a list $C_1 :: \cdots :: C_l$.

\newcommand{\inferwidth}{{\sf inferWidth}}
\newcommand{\inferscc}{{\sf inferSCC}}
\newcommand{\bab}{{\sf BaB}}
\newcommand{\maxfw}{{\sf maxFW}}
\newcommand{\sat}{{\sf SAT}}
\newcommand{\unsat}{{\sf UNSAT}}
\newcommand{\myvec}[1]{\overrightarrow{#1}}
\newcommand{\vars}{{\sf Vars}}

We define an inductive function $\inferwidth(cm, ls, S, \eta)$ to implement the width inference procedure, 
where $cm$ represents the $\fwc$-constraint $\varphi$, $ls$ represents a list 
of SCCs to be processed (initially it is a topological sort of the SCC graph of $G_\varphi$; in recursive calls it will be a prefix of this topological sort),
$S$ represents the set of variables whose widths have already been inferred, and $\eta: S \rightarrow \natnum$ represents the inferred widths.
Its result is $\unsat$ if the constraints are unsatisfiable,
and it is $(\sat, \theta)$ if the constraints are satisfiable;
then $\theta$ indicates the least solution of the variables in the SCCs in $ls$.
The function $\inferwidth(cm, ls, S, \eta)$ is implemented by the following  rules. 

\begin{itemize}
\item If $ls = \nil$, $\inferwidth(cm, ls, S, \eta)$ returns $(\sat, \emptyset)$. 

\item Otherwise, let $ls= ls'::C$,  $\inferwidth(cm, ls, S, \eta)$ calls $\inferscc(C, cm'_C)$ to infer the widths for the variables in $C$, where $cm'_C$ represents the $\Phi_W$-constraint corresponding to $C$, that is, $\dom(cm'_C) = V(C)$, and for each $x \in V(C)$, $cm'_C(x) = cm(x)[(\eta(z)/z)_{z \in S}]$. 
\begin{itemize}
\item If $\inferscc(C, cm'_C)$ returns $(\sat, \theta)$, then $\inferwidth(cm, ls', S', \eta')$ is called recursively, where $S' = S \cup V(C)$ and $\eta' = \eta \cup \theta$.  
\item Otherwise, $\inferwidth(cm, ls, S, \eta)$ returns $\unsat$. 
\end{itemize}
\end{itemize}
\vspace{-2mm}\begin{gather*} \small
   \label{solvefun-nil}
   \infer[\inferwidth: =\nil]
         {\seqq{\inferwidth(cm, ls, S, \eta) = (\sat,\emptyset)}}
         {\begin{array}{c}
         ls = \nil
           \end{array}
         }
\end{gather*}

\vspace{-2mm} \begin{gather*}\small
   \label{solvefun-not-nil-sat}
   \infer[\inferwidth: $\neq$\nil-{\sf sat-1}]
         {\seqq{\inferwidth(cm, ls,S,\eta) =(\sat, \theta \cup \eta'')}}
         {\begin{array}{c}
         ls = ls' :: C  \quad cm'_C = (x \rightarrow cm(x)[(\eta(z)/z)_{z \in S}])_{x \in V(C)}\\
          \inferscc(C, cm'_C) = (\sat, \theta) \quad S' = S \cup V(C) \quad \eta' = \eta \cup \theta \\
          \inferwidth(cm,ls',S',\eta') = (\sat, \eta'')
           \end{array}
         }
 \end{gather*} 
 
\vspace{-2mm} 
\begin{gather*}\small
   \label{solvefun-not-nil-unsat-1}
   \infer[\inferwidth: $\neq$\nil-{\sf unsat-1}]
         {\seqq{\inferwidth(cm, ls,S,\eta) = \unsat}}
         {\begin{array}{c}
         ls = ls' :: C  \quad cm'_C = (x \rightarrow cm(x)[(\eta(z)/z)_{z \in S}])_{x \in V(C)}\\
          \inferscc(C, cm'_C) = (\sat, \theta) \quad S' = S \cup V(C) \quad \eta' = \eta \cup \theta \\
         \inferwidth(cm,ls',S',\eta') = \unsat
           \end{array}
         }
 \end{gather*} 
 \vspace{-2mm}\begin{gather*}\small
   \label{solvefun-not-nil-unsat-2}
   \infer[\inferwidth: $\neq$\nil-{\sf unsat-2}]
         {\seqq{\inferwidth(cm, ls,S,\eta) = \unsat}}
         {\begin{array}{c}
         ls = ls' :: C  \quad cm'_C = (x \rightarrow cm(x)[(\eta(z)/z)_{z \in S}])_{x \in V(C)}\\
          \inferscc(C, cm'_C) = \unsat 
           \end{array}
         }
 \end{gather*} 
 
The function ${\inferscc}(C,cm)$ infers the widths of variables in a strongly connected component $C$ that needs to satisfy the {\firlis}-constraint $cm$.
Similar to above, its result is $\unsat$ if the constraints are unsatisfiable,
and it is $(\sat, \theta)$ if the constraints are satisfiable;
then $\theta$ indicates the least solution.
The function distinguishes whether $C$ is trivial or not.
\begin{itemize}
\item If $C$ is trivial, with $V(C) = \{ x \}$, {\inferscc} computes the assignment $\theta$ such that $\dom(\theta) = \{x\}$ and $\theta(x)$ is the maximum of $0$ and the constants $c$ in the inequalities $x \ge c \in cm(x)$. In this case, {\inferscc} returns $(\sat, \theta)$.

\item  If $C$ is not trivial, we further distinguish whether $C$ is expansive or not.
\begin{itemize}
\item If $C$ is expansive, {\inferscc} computes the upper bounds $\myvec{ub}$ and the lower bounds $\myvec{lb}$ and calls the function $\bab(cm, \myvec{lb}, \myvec{ub})$. (See the rule {\sf \inferscc: nontrivial-bab}, where the function ${\sf UB}(C, cm)$ computes the upper bounds $\myvec{ub}$ as described in Section~\ref{sec-procedure} and its detailed definition is omitted.) 
\item If $C$ is nonexpansive, {\inferscc} calls $\maxfw(C)$ to compute the maximum distances $md$ between vertices. If a cycle of positive weight is discovered, then {\inferscc} reports {\unsat}. Otherwise, it computes the least solution $\theta$. (See the rules \inferscc: {\sf nontrivial-maxfw-unsat} and \inferscc: {\sf nontrivial-maxfw-sat}.)
\end{itemize}
\end{itemize}
 \vspace{-2mm} \begin{gather*}\small
   \label{inferSCC-single}
   \infer[\inferscc: {\sf trivial}]
         {\seqq{\inferscc(C, cm) = (\sat, \theta)}}
         {\begin{array}{c}
         C \mbox{ is trivial} \quad V(C) = \{x\}\quad
         \theta(x) = \max\{0,c \mid x \ge c\in cm(x)\}
           \end{array}
         }
 \end{gather*}
\vspace{-2mm} \begin{gather*}\small
   \label{inferSCC-bab}
   \infer[\inferscc: {\sf nontrivial-bab}]
         {\seqq{\inferscc(C, cm) = \bab(cm, \myvec{lb}, \myvec{ub})}}
         {\begin{array}{c}
        C \mbox{ is expansive} \quad \myvec{ub} = {\sf UB}(C, cm) \\
          \myvec{lb} = \Big(\big(\max\big\{0, c \mid x \ge c + \sum_{x_j \in C} a_j x_j  \in cm(x)\big\}\big)_{x \in V(C)}\Big)
           \end{array}
         }
 \end{gather*}
\vspace{-2mm} \begin{gather*}\small
   \label{inferSCC-bab}
   \infer[\inferscc: {\sf nontrivial-maxfw-unsat}]
         {\seqq{\inferscc(C, cm) = \unsat}}
         {\begin{array}{c}
        C \mbox{ is nonexpansive} \quad md = \maxfw(C)\\
        \exists x \in V(C). \ md(x,x) > 0
           \end{array}
         }
 \end{gather*} 
\vspace{-2mm}   \begin{gather*}\small
   \label{inferSCC-bab}
   \infer[\inferscc: {\sf nontrivial-maxfw-sat}]
         {\seqq{\inferscc(C, cm) = (\sat, \theta)}}
         {\begin{array}{c}
        C \mbox{ is nonexpansive} \quad md = \maxfw(C) \\
        \forall x \in V(C).\ md(x,x) \le 0\\
        \myvec{v} = \Big(\big(\max\big\{0,c \mid x \ge c  \in cm(x)\big\}\big)_{x \in V(C)}\Big)\\
         \theta = \Big(\big(\max\big\{0, md(x,x')+v(x') \mid x' \in C\big\}\big)_{x \in V(C)}\Big)
           \end{array}
         }
 \end{gather*}
 

The function $\bab(cm, \myvec{lb}, \myvec{ub})$ infers the widths of the variables in an expansive SCC $C$. The Rocq implementation of the function $\bab$ is specified by the $\bab$ rules, which follow closely the description of the branch-and-bound procedure in Section~\ref{sec-procedure}. 
Note that in the rules for $\bab$, to simplify notation, the SCC $C$ is not introduced as a parameter of $\bab$,
but is assumed to be constant with $V(C) = \{x_1, \ldots, x_n\}$.
Also, for the map of constraints $cm$, we define $\psi_{cm} = \bigwedge_{x \in \dom(cm)} \bigwedge_{x \ge t \in cm(x)} x \ge t$.


\begin{gather*}\small
   \label{bab-lb-ub}
   \infer[{\sf \bab: not lb $\le$ ub}]
         {\seqq{\bab(cm, \myvec{lb}, \myvec{ub}) = \unsat}}
         {\begin{array}{c}
        \exists i \in [n].\ lb_i > ub_i 
           \end{array}
         }
 \end{gather*}

\vspace{-2mm}   
\begin{gather*}\small
   \label{bab-eq-sat}
   \infer[{\sf \bab: eq-sat}]
         {\seqq{\bab(cm, \myvec{lb}, \myvec{ub}) = (\sat,\myvec{lb})}}
         {\begin{array}{c}
         \myvec{lb} = \myvec{ub} \quad \psi_{cm}[lb_1/x_1, \ldots, lb_n/x_n] \mbox{ is } \ltrue
           \end{array}
         }
 \end{gather*}

\vspace{-2mm}   
\begin{gather*}\small
   \label{bab-eq}
   \infer[{\sf \bab: eq-unsat}]
         {\seqq{\bab(cm, \myvec{lb}, \myvec{ub}) = \unsat}}
         {\begin{array}{c}
         \myvec{lb} = \myvec{ub} \quad \psi_{cm}[lb_1/x_1, \ldots, lb_n/x_n] \mbox{ is } \lfalse
           \end{array}
         }
 \end{gather*}

\vspace{-2mm}   \begin{gather*}\small
   \label{bab-neq-true}
   \infer[{\sf \bab: neq-true}]
         {\seqq{\bab(cm, \myvec{lb}, \myvec{ub}) = \bab(cm, \myvec{lb}, \myvec{ub'})}}
         {\begin{array}{c}
         \myvec{lb} \lneqq \myvec{ub} \quad
         \myvec{v} = (\floor{(ub_1+lb_1)/2}, \ldots,  \floor{(ub_n+lb_n)/2})\\
        \psi_{cm}[v_1/x_1, \ldots, v_n/x_n]  \mbox{ is }  \ltrue \quad \myvec{ub'} = \myvec{v}
           \end{array}
         }
 \end{gather*}

\vspace{-2mm}   \begin{gather*}\small
   \label{bab-neq-false-rhsub}
   \infer[{\sf \bab: neq-false-rhs-1}]
         {\seqq{\bab(cm, \myvec{lb}, \myvec{ub}) = (\sat, \theta)}}
         {\begin{array}{c}
         \myvec{lb} \lneqq \myvec{ub} \quad
         \myvec{v} = (\floor{(ub_1+lb_1)/2}, \ldots, \floor{(ub_n+lb_n)/2})\\
         \psi_{cm}[v_1/x_1, \ldots, v_n/x_n] \mbox{ is } \lfalse\\
         x_i \ge a_0+\sum \limits_{j \in [n]} a_j x_j \in cm(x_i) \mbox{ and } v_i < a_0+\sum \limits_{j \in [n]} a_j v_j  \\ 
         j_0 \in [n] \quad a_{j_0} >0 \quad lb_{j_0}<v_{j_0}\\
         \myvec{ub'} = (ub_1, \ldots, ub_{j_0-1}, v_{j_0} -1, ub_{j_0+1}, \ldots, ub_n) \\
         \bab(cm, \myvec{lb}, \myvec{ub'}) = (\sat, \theta)
           \end{array}
         }
 \end{gather*}
 
\vspace{-2mm}   \begin{gather*}\small
   \label{bab-neq-false-lhs}
   \infer[{\sf \bab: neq-false-rhs-2}]
         {\seqq{\bab(cm, \myvec{lb}, \myvec{ub}) = \bab(cm, \myvec{lb'}, \myvec{ub})}}
         {\begin{array}{c}
         \myvec{lb} \lneqq \myvec{ub} \quad 
         \myvec{v} = (\floor{(ub_1+lb_1)/2}, \ldots, \floor{(ub_n+lb_n)/2})\\
        \psi_{cm}[v_1/x_1, \ldots, v_n/x_n] \mbox{ is } \lfalse\\
         x_i \ge a_0+\sum \limits_{j \in [n]} a_j x_j \in cm(x_i) \mbox{ and } v_i < a_0+\sum \limits_{j \in [n]} a_j v_j  \\ 
         j_0 \in [n] \quad a_{j_0} >0 \quad lb_{j_0}<v_{j_0}\\
         \myvec{ub'} = (ub_1, \ldots, ub_{j_0-1}, v_{j_0} -1, ub_{j_0+1}, \ldots, ub_n)\\
         \bab(cm, \myvec{lb}, \myvec{ub'}) = \unsat \\
         \myvec{lb'} = (lb_1, \ldots, lb_{j_0-1}, v _{j_0}, lb_{j_0+1}, \ldots, lb_n)
           \end{array}
         }
 \end{gather*}
 
\begin{gather*}\small
   \label{bab-neq-false-lhs}
   \infer[{\sf \bab: neq-false-lhs}]
         {\seqq{\bab(cm, \myvec{lb}, \myvec{ub}) = \bab(cm, \myvec{lb'}, \myvec{ub})}}
         {\begin{array}{c}
         \myvec{lb} \lneqq \myvec{ub} \quad
         \myvec{v} = (\floor{(ub_1+lb_1)/2}, \ldots, \floor{(ub_n+lb_n)/2})\\
         \psi_{cm}[v_1/x_1, \ldots, v_n/x_n] \mbox{ is } \lfalse\\
         x_i \ge a_0+\sum \limits_{j \in [n]} a_j x_j \in cm(x_i) \mbox{ and } v_i < a_0+\sum \limits_{j \in [n]} a_j v_j  \\ 
         \forall j_0 \in [n].\ a_{j_0} >0 \rightarrow lb_{j_0}=v_{j_0}\\
         \myvec{lb'} = (lb_1, \ldots, lb_{i-1}, v_i+1, lb_{i+1}, \ldots, lb_n)
           \end{array}
         }
 \end{gather*}
 



\subsection{Specification of the width inference procedure}

We first introduce some notations. 
For $ls = C_1 :: \cdots :: C_k$, we use $V(ls)$ to denote $\bigcup_{j \in [k]} V(C_j)$. 
Moreover, we define
\[
\psi_{cm, ls, \eta} \eqdef \bigwedge \limits_{x \in V(ls)} \bigwedge \limits_{x \ge t \in cm(x)} x \ge t[(\eta(x')/x')_{x' \in S}].
\]
Note that $S \cap V(ls) = \emptyset$ and $S \cup V(ls) = \dom(cm)$.
If $ls = \nil$, then $\psi_{cm, ls, \eta} = \ltrue$ by convention.  

The correctness of $\inferwidth(cm, ls, S, \eta)$ is specified by $P_{\inferwidth}$, which asserts that if $\inferwidth$ returns $(\sat, \eta')$, then $\eta'$ is the least solution of $\psi_{cm, ls,\eta}$, otherwise (i.e., $\inferwidth$ returns $\unsat$),  $\psi_{cm, ls,\eta}$ is unsatisfiable. 
\[\small
\begin{array}{l}
P_{\inferwidth} \eqdef \forall cm, ls, S, \eta, \eta'.\\
\qquad\qquad\left(\begin{array}{l}
 \left(
\begin{array}{l}
\inferwidth(cm, ls, S, \eta) = (\sat,\eta') \rightarrow \\
\left(\begin{array}{l}\psi_{cm, ls, \eta}[(\eta'(x)/x)_{x \in V(ls)}]\ \wedge \\
\forall \eta''. \left( \psi_{cm, ls, \eta}[(\eta''(x)/x)_{x \in V(ls)}] \rightarrow \eta' \leq \eta'' \right)\end{array}\right)
\end{array}
\right)\ \bigwedge\\
\inferwidth(cm, ls, S, \eta) = \unsat \rightarrow \forall \eta''.\  \neg \psi_{cm, ls, \eta}[(\eta''(x)/x)_{x\in V(ls)}]
\end{array}\right).
\end{array}
\]
Note that $\eta, \eta', \eta''$ in $P_{\inferwidth}$ 
are quantified (second-order) variables representing assignment functions, instead of concrete assignments. 

The proof that $\inferwidth(cm, ls, S, \eta)$ satisfies $P_\inferwidth$ is done by an induction on $ls$, a list of SCCs in $G_\varphi$. 
Moreover, the proof relies on the proof that $\inferscc$ satisfies the property $P_{\inferscc}$ defined in the sequel,
\[\small
P_{\inferscc} \eqdef
\forall cm, \theta. \left(
\begin{array}{l}
 \left(
\begin{array}{l}
\inferscc(C, cm) = (\sat,\theta) \rightarrow \\
\left(\begin{array}{l}\psi_{cm} [(\theta(x)/x)_{x \in \dom(cm)}]\ \wedge \\
\forall \theta'. \left( \psi_{cm} [(\theta'(x)/x)_{x \in \dom(cm)}]\rightarrow \theta \leq \theta'
\right)\end{array}\right)
\end{array}
\right)\ \bigwedge\\
\inferscc(C, cm) = \unsat \rightarrow \forall \theta'.\ \neg \psi_{cm}[(\theta'(x)/x)_{x \in \dom(cm)}]
\end{array}\right).
\]

In turn, the proof that $\inferscc$ satisfies $P_{\inferscc}$ relies on the proof that $\bab(cm, \myvec{lb}, \myvec{ub})$ and $\maxfw(C)$ satisfy the correctness property $P_{\bab}$ and $P_{\maxfw}$ respectively. 
\[\small
\begin{array}{l}
P_{\bab} \eqdef \forall cm, \myvec{lb}, \myvec{ub}, \theta. \\
\qquad\left(
\begin{array}{l}
\left(
\begin{array}{l}
\bab(cm, \myvec{lb}, \myvec{ub}) = (\sat,\theta) \rightarrow \\
\left(\begin{array}{l}\psi_{cm} [(\theta(x)/x)_{x \in \dom(cm)}] \wedge \myvec{lb} \le \theta \le \myvec{ub} \ \wedge \\
\forall \theta'.\big(\psi_{cm}[(\theta'(x)/x)_{x \in \dom(cm)}] \wedge \myvec{lb} \le \theta' \le \myvec{ub} \rightarrow \theta \leq \theta'
\big)\end{array}\right)
\end{array}
\right)\ \bigwedge\\
\bab(cm, \myvec{lb}, \myvec{ub}) = \unsat \rightarrow \forall \theta'.\ \neg \big(\psi_{cm}[(\theta'(x)/x)_{x \in \dom(cm)}] \wedge \myvec{lb} \le \theta' \le \myvec{ub}\big)
\end{array} \right).
\end{array}
\]

\[\small
\begin{array}{l}
P_{\maxfw} \eqdef 
\forall C, md.\ \maxfw(C) = md\ \rightarrow \\
\qquad\left( 
\begin{array}{l}
\left( 
\begin{array}{l}
(\forall x \in V(C).\ md(x,x) \le 0)\ \rightarrow\\
\forall x,x' \in V(C).\ \forall \mbox{ path } \pi \mbox{ from } x \mbox{ to } x'.\ md(x, x') \ge weight(\pi)
\end{array}
\right)\ \bigwedge\\
\left(
\begin{array}{l}
(\exists x \in V(C). \ md(x,x) > 0) \ \rightarrow \\
\exists x \in V(C).\ \exists \mbox{ a simple cycle } \pi \mbox{ from }  x \mbox{ to } x.\ weight(\pi) > 0
\end{array}
\right)
\end{array}\right),
\end{array}
\]
where $weight(\pi)$ denotes the weight of $\pi$, that is, the sum of the weights of the edges therein. 



\subsection{Verification of the width inference procedure}

The inductive proof that $\inferwidth(cm, ls, S, \eta)$ satisfies $P_{\inferwidth}$ proceeds as follows. 

\smallskip\noindent
{\bf Basis step $ls=\nil$.} Then $\inferwidth(cm, ls, S, \eta) = (\sat, \emptyset)$. Then $\psi_{cm, ls, \eta} = \ltrue$. It is easy to check that $P_{\inferwidth}$ is $\ltrue$ in this case.  

\smallskip\noindent
{\bf Induction step.} Suppose the result holds for $ls'$ of length less than or equal to $n$. Let us consider $ls = ls'::C$ of length $n+1$. 
\begin{itemize}
    \item If $\inferwidth(cm, ls, S, \eta) = (\sat, \eta')$, then from the rule {\inferwidth: $\neq$\nil-{\sf sat}}, we know that $\inferscc(C, cm'_C) = (\sat, \theta)$ and $\inferwidth(cm, \linebreak[0] ls', \linebreak[0] S \cup V(C),\linebreak[0] \eta \cup \theta) = (\sat, \eta'')$ for some $\theta$ and $\eta''$. By the induction hypothesis, $\eta''$ is the least solution of $\psi_{cm, ls', \eta \cup \theta}$. Moreover, from $P_{\inferscc}$, we know that $\theta$ is the least solution of $\psi_{cm'_C}$. It is not hard to see that $\psi_{cm, ls, \eta} \models \psi_{cm'_C} \wedge \psi_{cm, ls', \eta \cup \theta}$. From the fact that $\theta \cup \eta''$ is the least solution of $\psi_{cm'_C} \wedge \psi_{cm, ls', \eta \cup \theta}$, we deduce that $\theta \cup \eta'' = \eta'$ is the least solution of $\psi_{cm, ls, \eta}$.
    
    \item If $\inferwidth(cm, ls, S, \eta) = \unsat$,  then either $\inferscc(C, cm'_C) = \unsat$ or ($\inferscc(C, cm'_C) = (\sat,\theta)$ and $\inferwidth(cm, ls', S \cup C, \eta \cup \theta) = \unsat$). 
    \begin{itemize}
    \item If $\inferscc(C, cm'_C) = \unsat$, then from $P_{\inferscc}$, we know that $\psi_{cm'_C}$ is unsatisfiable. From the fact that $\psi_{cm, ls, \eta} \models \psi_{cm'_C}$, we deduce that $\psi_{cm, ls, \eta}$ is unsatisfiable as well.  
    \item If $\inferscc(C, cm'_C) = (\sat,\theta)$ and $\inferwidth(cm, ls', S \cup C, \eta \cup \theta) = \unsat$, then by the induction hypothesis, $\psi_{cm, ls', \eta \cup \theta}$ is unsatisfiable. Moreover, from $P_{\inferscc}$, $\theta$ is the least solution of $\psi_{cm'_C}$. It is not hard to see that $\psi_{cm, ls, \eta} \models \psi_{cm'_C} \wedge \psi_{cm, ls', \eta \cup \theta}$. From this fact, we deduce that $\psi_{cm, ls, \eta}$ is unsatisfiable. 
    \end{itemize}
\end{itemize}

The proof that $\inferscc(C, cm)$ satisfies $P_{\inferscc}$ is done by a case analysis. 
\begin{itemize}
\item If $C$ is trivial, then every inequality in $cm(x)$ is of the form $x \ge c$.  As a result, the least (nonnegative) solution is $\max\{0, c \mid x \ge c \in cm(x)\}$.

\item If $C$ is nontrivial and expansive, then from the rule $\inferscc \mbox{: \sf nontrivial-bab}$,  the proof mainly relies on the proof that $\bab(cm, \myvec{lb}, \myvec{ub})$ satisfies $P_{\bab}$, which can be done by following the branch-and-bound procedure and applying an induction on the recursive calls of $\bab$. 

\item If $C$ is nontrivial and nonexpansive, then we use the rules $\inferscc:{}${\sf nontrivial-maxfw-unsat} and $\inferscc:{}${\sf nontrivial-maxfw-sat}
and mainly rely on the proof that the maximal Floyd--Warshall algorithm $\maxfw(C)$ satisfies $P_{\maxfw}$, which essentially follows the proof of the correctness of Floyd--Warshall in Isabelle~\cite{Wimmer_Lammich_AFP_FloydWarshall}. 
\end{itemize}

We now elaborate how to prove that $\bab(cm, \myvec{lb}, \myvec{ub})$ satisfies $P_{\bab}$.

\begin{itemize}
    \item If there exists $i \in [n]$ such that $lb_i > ub_i$, then according to the rule {\sf \bab: not lb $\le$ ub}, $\bab(cm, \myvec{lb}, \myvec{ub})$ returns $\unsat$. Then evidently $\psi_{cm}(\myvec{x}) \wedge \myvec{lb} \le \myvec{x} \le \myvec{ub}$ is unsatisfiable. Therefore, in this case, $P_{\bab}$ holds. In the sequel, let us assume that $\forall i \in [n].\ lb_i \le ub_i$. 

\hide{
\begin{gather*}\small
   \label{bab-lb-ub}
   \infer[{\sf \bab: not lb $$\le$ ub}]
         {\seqq{\bab(cm, \myvec{lb}, \myvec{ub}) = \unsat}}
         {\begin{array}{c}
        \exists i \in [n].\ lb_i > ub_i 
           \end{array}
         }
 \end{gather*}
}

    %
    \item If $\myvec{lb} = \myvec{ub}$,  then $\myvec{lb}$ is the only possible solution of $\psi_{cm}(\myvec{x}) \wedge \myvec{lb} \le \myvec{x} \le \myvec{ub}$. 
    \begin{itemize}
    \item If $\psi_{cm}[lb_1/x_1, \ldots, lb_n/x_n]$ is $\ltrue$, then according to the rule {\sf \bab: eq-sat}, $\bab(cm, \myvec{lb}, \myvec{ub})$ returns $(\sat, \myvec{lb})$. Evidently, $\myvec{lb}$ is the least solution of $\psi_{cm}(\myvec{x}) \wedge \myvec{lb} \le \myvec{x} \le \myvec{ub}$, thus $P_{\bab}$ holds in this case.
    \item Otherwise, $\psi_{cm}[lb_1/x_1, \ldots, lb_n/x_n]$ is $\lfalse$. According to the rule {\sf \bab: eq-unsat}, we know that $\bab(cm, \myvec{lb}, \myvec{ub})$ returns $\unsat$. In this case, evidently, $\psi_{cm}(\myvec{x}) \wedge \myvec{lb} \le \myvec{x} \le \myvec{ub}$ is unsatisfiable. Therefore, $P_{\bab}$ holds in this case. 
    \end{itemize}

\hide{
\begin{gather*}\small
   \label{bab-eq-sat}
   \infer[{\sf \bab: eq-sat}]
         {\seqq{\bab(cm, \myvec{lb}, \myvec{ub}) = (\sat,\myvec{lb})}}
         {\begin{array}{c}
         \myvec{lb} = \myvec{ub} \quad \psi_{cm}[lb_1/x_1, \ldots, lb_n/x_n] \mbox{ is } \ltrue
           \end{array}
         }
 \end{gather*}

\begin{gather*}\small
   \label{bab-eq}
   \infer[{\sf \bab: eq-unsat}]
         {\seqq{\bab(cm, \myvec{lb}, \myvec{ub}) = \unsat}}
         {\begin{array}{c}
         \myvec{lb} = \myvec{ub} \quad \psi_{cm}[lb_1/x_1, \ldots, lb_n/x_n] \mbox{ is } \lfalse
           \end{array}
         }
 \end{gather*}
}

     \item If $\myvec{lb} \neq \myvec{ub}$ and $\psi_{cm}[v_1/x_1,\ldots,v_n/x_n]$ is evaluated to $\ltrue$, then according to the rule {\sf \bab: neq-true}, $\bab(cm, \myvec{lb}, \myvec{ub})$ calls $\bab(cm, \myvec{lb}, \myvec{v})$ recursively and takes its output as the output.
     It is easy to see that $\myvec{v} < \myvec{ub}$.
     Then by the induction hypothesis, $\bab(cm, \myvec{lb}, \myvec{v})$ returns the least solution of of $\psi_{cm}(\myvec{x}) \wedge \myvec{lb} \le \myvec{x} \le \myvec{v}$ (since a solution $\myvec{v}$ already exists). 
     From $\myvec{v} < \myvec{ub}$, the least solution of $\psi_{cm}(\myvec{x}) \wedge \myvec{lb} \le \myvec{x} \le \myvec{v}$ is also the least solution of $\psi_{cm}(\myvec{x}) \wedge \myvec{lb} \le \myvec{x} \le \myvec{ub}$.
     Therefore, $P_{\bab}$ holds in this case.

\hide{
\vspace{-2mm}   \begin{gather*}\small
   \label{bab-neq-true}
   \infer[{\sf \bab: neq-true}]
         {\seqq{\bab(cm, \myvec{lb}, \myvec{ub}) = \bab(cm, \myvec{lb}, \myvec{ub'})}}
         {\begin{array}{c}
         \myvec{lb} \neq \myvec{ub} \quad
         \myvec{v} = (\floor{(ub_1+lb_1)/2}, \ldots,  \floor{(ub_n+lb_n)/2})\\
        \psi_{cm}[v_1/x_1, \ldots, v_n/x_n] \ is\ \ltrue \quad \myvec{ub'} = \myvec{v}
           \end{array}
         }
 \end{gather*}
}

    \item If $\myvec{lb} \neq \myvec{ub}$ and $\psi_{cm}[v_1/x_1,\ldots,v_n/x_n]$ is evaluated to $\lfalse$, then there is an inequality $x_i \ge a_0 + \sum_{j \in [n]} a_j x_j$ for some $i \in [n]$ in $\psi_{cm}$ such that 
    $v_i > a_0 + \sum_{j \in [n]} a_j v_j$. We distinguish between the following cases. 
    \begin{itemize}
        \item If $\exists j_0\in [n].\ v_{j_0}>0 \land a_{j_0}>0 \land lb_{j0}<v_{j_0}$, then according to the rule {\sf \bab: neq-false-rhs-1} and {\sf \bab: neq-false-rhs-2}, $\bab(cm, \myvec{lb}, \myvec{ub})$ recursively calls $\bab(cm,\myvec{lb},\myvec{ub'})$, where $\myvec{ub'} = (ub_1, \ldots, ub_{j_0-1}, v_{j_0} -1, ub_{j_0+1}, \ldots, ub_n)$. 
        \begin{itemize}
        \item If $\bab(cm,\myvec{lb},\myvec{ub'}) = (\sat, \theta)$, then by the induction hypothesis, $\theta$ is the least solution of $\psi_{cm}(\myvec{x}) \wedge \myvec{lb} \le \myvec{x} \le \myvec{ub'}$. 
        Because $\myvec{ub'} \le \myvec{ub}$, it follows that $\theta$ is also the least solution of  $\psi_{cm}(\myvec{x}) \wedge \myvec{lb} \le \myvec{x} \le \myvec{ub}$.
        From the rule {\sf \bab: neq-false-rhs-1}, $\bab(cm,\myvec{lb},\myvec{ub}) = (\sat, \theta)$.
        Therefore, $P_{\bab}$ holds in this case. 
        \item Otherwise, $\bab(cm, \myvec{lb}, \myvec{ub'}) = \unsat$.
        Then by the induction hypothesis, $\psi_{cm}(\myvec{x}) \wedge \myvec{lb} \le \myvec{x} \le \myvec{ub'}$ is unsatisfiable.
        As a result, if $\psi_{cm}(\myvec{x}) \wedge \myvec{lb} \le \myvec{x} \le \myvec{ub}$ is satisfiable, then we must have $x_{j_0} \ge v_{j_0}$.
        Therefore, the least solution of $\psi_{cm}(\myvec{x}) \wedge \myvec{lb} \le \myvec{x} \le \myvec{ub}$ (if it exists) is that of $\psi_{cm}(\myvec{x}) \wedge \myvec{lb'} \le \myvec{x} \le \myvec{ub}$, where 
        $$\myvec{lb'} = (lb_1, \ldots, lb_{j_0-1}, v _{j_0}, lb_{j_0+1}, \ldots, lb_n).$$ 
        From the rule {\sf \bab: neq-false-rhs-2}, $\bab(\psi_{cm},\myvec{lb},\myvec{ub})$ recursively calls $\bab(\psi_{cm}, \myvec{lb'}, \myvec{ub})$.
        By the induction hypothesis, if $\bab(\psi_{cm}, \myvec{lb'}, \myvec{ub})$ returns $(\sat, \theta)$,
        then $\theta$ is the least solution of $\psi_{cm}(\myvec{x}) \wedge \myvec{lb'} \le \myvec{x} \le \myvec{ub}$,
        otherwise, $\psi_{cm}(\myvec{x}) \wedge \myvec{lb'} \le \myvec{x} \le \myvec{ub}$ is unsatisfiable.
        Therefore, $P_{\bab}$ holds in this case. 
        \end{itemize}
        \item If $\forall j_0\in [n].\  a_{j_0}>0 \rightarrow lb_{j0} = v_{j_0}$, then it is impossible to decrease the right-hand side of $x_i \ge a_0 + \sum_{j \in [n]} a_j x_j$ by choosing values smaller than $\myvec{v}$.
        Therefore, it is necessary to increase the value of $x_i$, that is, it must be the case that $x_i \ge v_i+1$. 
        So, the least solution of $\psi_{cm}(\myvec{x}) \wedge \myvec{lb} \le \myvec{x} \le \myvec{ub}$ (if it exists) is that of $\psi_{cm}(\myvec{x}) \wedge \myvec{lb'} \le \myvec{x} \le \myvec{ub}$,
        where $\myvec{lb'} = (lb_1, \ldots, lb_{i-1}, v_i + 1, lb_{i+1}, \ldots, lb_n)$.
        According to the rule {\sf \bab: neq-false-lhs}, in this case, $\bab(cm, \myvec{lb}, \myvec{ub})$ recursively calls $\bab(cm, \myvec{lb'}, \myvec{ub})$.
        By the induction hypothesis, if $\bab(cm, \myvec{lb'}, \myvec{ub})$ returns $(\sat, \theta)$,
        then $\theta$ is the least solution of $\psi_{cm}(\myvec{x}) \wedge \myvec{lb'} \le \myvec{x} \le \myvec{ub}$,
        otherwise, $\psi_{cm}(\myvec{x}) \wedge \myvec{lb'} \le \myvec{x} \le \myvec{ub}$ is unsatisfiable.
        As a result, in this case, $P_{\bab}$ holds. 
    \end{itemize}
\end{itemize}

\section{OCaml implementation and evaluation}\label{sec-impl-eval}

\subsection{Implementation} 
From the Rocq implementation of the inductive function $\inferwidth$ in Section~\ref{sec-coq-proof}, we utilized the extraction mechanism available in the Rocq proof assistant~\cite{Letouzey02,LET08} to extract OCaml code. 
Moreover, as the $\inferwidth$ function takes a $\Phi_W$-constraint as input, we implemented a parser to translate a {\firrtl} program into its abstract syntax tree (AST), from which
the $\Phi_W$-constraint is extracted
by a constraint extractor.
The parser is implemented by {\tt ocamlyacc}~\cite{ocamlyacc}, where a context-free grammar of the {\firrtl} language (with attached semantic actions) is specified, while 
the constraint extractor is extracted from the Rocq implementation. 
The OCaml implementation of the $\inferwidth$ function, together with the parser and constraint extractor, constitutes an OCaml implementation of the InferWidths pass of  {\firrtl} (see Fig.~\ref{fig-tool-arch}), which we call BFWInferWidths (``BFW'' is an abbreviation of ``Branch-and-bound + Floyd--Warshall''). 

\begin{figure}[!h]
\vspace{-4mm}
\centering
\includegraphics[scale = 0.82]{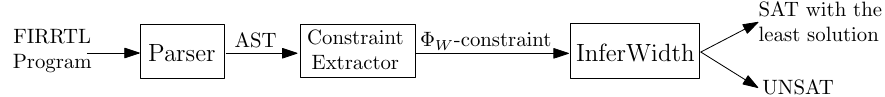} 
\caption{The architecture of BFWInferWidths}\label{fig-tool-arch}\vspace{-4mm}
\end{figure}

When running BFWInferWidths on the width inference instance generated from 
the RISC-V BOOM processor, the non-tail-recursive implementations of the functions {\sf split}, {\sf concat} and {\sf flatmap} in the OCaml standard library {\sf Base} caused
segmentation faults and stack overflows.
The main reason is that they do not scale to large real-world instances, which may have more than 200,000 variables and 
more than 130,000 inequalities. 
We replaced these non-tail-recursive implementations  
with tail-recursive versions. 
\subsection{Evaluation}\label{sec-evaluation}
 
\subsubsection*{Benchmarks.} 
To evaluate the performance of BFWInferWidths, we collected 75~{\firrtl} programs that contain  components with unspecified widths:
\begin{itemize}
\item 4 {\firrtl} programs in the specification of {\firrtl}~\cite{firrtl-spec},  
\item 20 {\firrtl} programs generated from Chisel programs in the Chisel book~\cite{chisel:book},  
\item 11 {\firrtl} programs used in the unit tests for {\firtool}~\cite{firtool},
%
\item 26 {\firrtl} programs reported as issues for {\firtool}~\cite{firtool-issues}, 
\item 11 (manually generated) {\firrtl} programs that contain circular dependencies between the  widths-unspecified components,
\item 3 {\firrtl} programs for the production-grade Chisel-based RISC-V processor designs, NutShell~\cite{NutShell}, Rocket Chip~\cite{RocketChip} and RISC-V BOOM~\cite{Boom}.
\end{itemize}
The 75 {\firrtl} programs are partitioned  into two benchmark suites, i.e., REALWORLD (the 3 {\firrtl} programs of NutShell, Rocket Chip and RISC-V BOOM) and  MANUAL (all the other {\firrtl} programs.)

\revise{We found that \prettyll{dshl} occurs in three benchmarks, namely, NutShell (7 times),  Rocket Chip (3 times) and RISC-V BOOM (3955 times).
All occurrences belong to one of the following two cases.
\begin{itemize}
    \item The shift is actually not dynamic, but by a constant number of bits: For instance, \prettyll{y=dshl(x,3)}, whose  width constraint is $w_\mathtt{y} \geq w_\mathtt{x} + 2^3-1$, i.e., $w_\mathtt{y} \geq w_\mathtt{x} + 7$.
    \item There is no circular dependency: For instance, \prettyll{y=dshl(x, v)} with width constraint $w_\mathtt{y} \geq w_\mathtt{x} + 2^{w_\mathtt{v}}-1$, but the width of \prettyll{v} does not depend on the width of \prettyll{y}.
    Then $w_{\tt v}$ can be inferred before inferring $w_\mathtt{y}$ (in separate calls to the translation of $\inferscc$).
    Suppose $w_{\tt v}=c2$, then $w_\mathtt{y} \geq w_\mathtt{x} + 2^{w_{\tt v}}-1$ reduces to $w_\mathtt{y} \geq w_\mathtt{x} + 2^{c2}-1$, which is a linear constraint.
\end{itemize}
Thus, all \prettyll{dshl}'s occurring in the benchmarks can be handled by our procedure. However,  
we leave supporting general dynamic shift left \prettyll{dshl} as future work.}


\subsubsection*{Experimental setup.}
We compare the correctness and efficiency of BFWInferWidths with the implementation of the InferWidths pass in {\firtool}.
\revise{We also compare with the industrial ILP solver Gurobi~\cite{Gurobi}, that is a branch-and-bound based optimizer and represents the state-of-the-art according to~\cite{2024ParaILP}. Gurobi provides the $\min$ function, allowing us to encode the FIRWINE constraints and solve them 
using a single objective function that sums up all the width variables. Recall that if the problem has a solution, then it must have a unique least solution (cf. Proposition~\ref{prop-min-sol}). Thus,
a single objective function is sufficient instead of  multi-objective optimization (i.e., one objective function per width variable).}

\revise{We remark that the $\min$ function can be eliminated by using Big-M (a sufficiently large integer number) for ILP optimizers that do not support the $\min$ function. For instance, $a \geq \min(y,z)$ can be expressed by
$a \geq x \wedge y \geq x \wedge z \geq x \wedge x \geq y - M * b \wedge x \geq - M + z + M * b$,
where $x$ is an auxiliary non-negative variable and $b$ is a Boolean variable.
\hide{
Intuitively,
\begin{itemize}
    \item if $b=1$, then $a \geq x \wedge y \geq x \wedge z \geq x \wedge x \geq y - M \wedge x \geq z$, implying that $a \geq x = z = \min(y,z)$;
    \item if $b=0$, then  $a \geq x \wedge y \geq x \wedge z \geq x \wedge x \geq y \wedge x \geq -M + z$, implying that $a \geq x = y = \min(y,z)$.
\end{itemize}}
}

Note that although the correctness of our InferWidths was formally proven in Section~\ref{sec-coq-proof}, BFWInferWidths contains unverified components (the parser and constraint extractor), as well as the unverified OCaml standard library
{\sf Base}. Hence an empirical evaluation is useful to confirm the correctness of BFWInferWidths as a whole. 

All experiments were run on a 3.2 GHz 8-core Apple M1 laptop with 8 GB RAM. For the comparison, we use {\firtool} v1.73 and Gurobi v12.0.1 with 8-thread parallel optimization.
(Note that while Gurobi uses 8~threads,  BFWInferWidths and {\firtool} use only a single thread.)



\subsubsection*{Experimental results.} 
The results for evaluating correctness are reported in 
Table~\ref{tab:correctness-testing}.
We observe that the results of BFWInferWidths coincide to those of {\firtool} and Gurobi, except that {\firtool} failed on 12 instances in the MANUAL benchmark suite.  
Specifically, {\firtool} did not terminate in 24 hours on one {\firrtl} program (yielding 192 inequalities) which does not contain any circular dependencies, while both BFWInferWidths and Gurobi solved it in less than 10 ms. Indeed, it is due to
an unresolved issue in {\firtool}~\cite{firtoolIssue}. 
%
{\firtool} misthrew exceptions on the other 11 {\firrtl} programs due to circular dependencies.

\begin{table}[t]
\caption{Correctness of BFWInferWidths: compared with {\firtool} and Gurobi.}
\label{tab:correctness-testing}
\setlength{\tabcolsep}{8pt}
\centering
\begin{tabular}{|@{\;}c@{\;}|@{\;}c@{\;}|c|c|c|}
  \hline
  \multirow{2}{*}{\textbf{Benchmark suite}} & 
  \multirow{2}{*}{\textbf{\# programs}} & 
  \multicolumn{3}{c|}{\textbf{Number of solved instances}} \\ 
  \cline{3-5}
  & & {\firtool} & \textbf{Gurobi} & \textbf{BFWInferWidths} \\ 
  \hline
  MANUAL & 72 & 60 & 72 & 72 \\ 
  \hline
  REALWORLD & \phantom{0}3 & \phantom{0}3 & \phantom{0}3 & \phantom{0}3\\ 
  \hline
\end{tabular} 
\vspace{-2mm}
\end{table}    

The results for evaluating efficiency are reported in Table~\ref{tab:efficiency-conformancetesting},
where Column Avg. \#cpnts shows the average number of components declared with unspecified widths. We observe that on the MANUAL benchmark suite (where the number of components declared with unspecified widths 
is small), 
BFWInferWidths is at least 7 (resp. 10) times faster than {\firtool} (resp. Gurobi) on average. 
On the REALWORLD benchmark suite, 
BFWInferWidths is more efficient than {\firtool}, and is largely comparable with 
Gurobi (BFWInferWidths is faster on NutShell and Rocket Chip, but is slightly slower on RISC-V BOOM). 
Detailed results are reported in \arxiv{Appendix~\ref{app:exp}}\crv{\cite{ESOP26-full}}.

\begin{table}[t]
\setlength{\tabcolsep}{3pt}
\caption{Efficiency of BFWInferWidths compared with {\firtool} and Gurobi. (The time of {\firtool} is calculated by using the 60 instances that can be solved by it.)}
\label{tab:efficiency-conformancetesting}
\centering
\begin{tabular}{|cc|c|c|c|c|}
\hline
\multicolumn{2}{|c|}{\multirow{2}{*}{\textbf{Benchmark}}} & 
  \textbf{Avg.} & 
  \multicolumn{3}{c|}{\textbf{Time (ms) per instance}}\\ 
   \cline{4-6} 
  &  & \textbf{\#cpnts} & {\firtool} & \textbf{Gurobi} & \textbf{BFWInferWidths}\\ 
  \hline
  \multicolumn{2}{|c|}{MANUAL (72)} & \phantom{000,}119 & \phantom{0,00}7.49 & \phantom{0,0}12.07 & \fontseries{b}\selectfont\phantom{0,00}1.00 \\ 
  \hline
  \multicolumn{1}{|c|}{\multirow{3}{*} {REALWORLD}} &  NutShell & \phantom{00}7,152 & \phantom{0,}190.70 & \phantom{0,}194.55 & \fontseries{b}\selectfont\phantom{0,}158.31 \\ 
  \cline{2-6}
  \multicolumn{1}{|c|}{} & Rocket Chip & \phantom{00}4,882 & \phantom{0,}127.90 & \phantom{0,}120.64 & \fontseries{b}\selectfont\phantom{0,0}22.24 \\ 
  \cline{2-6}
   \multicolumn{1}{|c|}{}  & RISC-V BOOM & 205,608 & 8,338.30 & \fontseries{b}\selectfont 3,326.94 & 3,467.80 \\ 
  \hline
\end{tabular} 
\vspace{-2mm}
\end{table}

\revise{We also have performed a preliminary comparison with Z3’s optimizer~\cite{Programming-Z3}, by encoding the width inference problem as the minimization problem. On the BOOM benchmark, our approach is approximately $26 \times $ faster than Z3. Crucially, like Gurobi, Z3 does not provide a formally verified solver, which is the key distinction we should emphasize.}

To summarize, as an OCaml program that is extracted from the formally verified Coq implementation, BFWInferWidths demonstrates superior performance than {\firtool},
the state-of-the-art compiler for {\firrtl}.
It is comparable to Gurobi, the industrial ILP solver, executed with 8~threads.

\section{Related work}\label{sec-related}


\subsubsection*{Verification of compilers.} 
Compiler verification dates back to 1967,
when \cite{MP67} proved the correctness of a translation from 
arithmetic expressions to stack
machine code. It has become a vibrant area 
ranging from
single compilation pass to sophisticated code optimizations in compilers~\cite{Dave03}.

Representative works on verified compilation of software programming languages
include CompCert for C~\cite{Leroy09,kastner2018compcert,Leroy2006,Leroy-BKSPF-2016,BlazyDL06,Kastner-LBSSF-2017},
V\'{e}lus for Lustre~\cite{BourkeBDLPR17}, Ve\-llvm~\cite{ZhaoNMZ12,ZhaoNMZ13} and Cre\-llvm~\cite{KangKSLPSKCCHY18} for LLVM IR,
CakeML for ML~\cite{KumarMNO14,LoowKTMNAF19,TanMKFON19,OwensNKMT17}, to cite a few.
In particular, the interactive theorem prover Rocq has been widely adopted to formalize semantics of programs,
compilation and optimization passes; on its basis compilation correctness can be proven, and
executable inter\-preters/\linebreak[0]com\-pilers can be automatically synthesized (cf.~\cite{Awesomecoq}).
For instance, \cite{benton2009} formalized and proved a compiler that translates programs in a simple functional language into code in an idealized assembly language; 
\cite{Chlipala07,Chlipala10} formalized and proved compilers that translate programs in simply-typed or untyped lambda calculus to code in an idealized assembly language; 
CompCert is a verified compiler that translates programs in Clight (a large subset of C) to code in PowerPC assembly language.
For domain-specific systems, 
\cite{GuillaumeJ24} designed a safety-oriented language for low-level algebra libraries. Their CompCert-based compiler formally preserves semantics (verified in Rocq), while the language enables deductive verification of safety properties.
\cite{NoelENHDGZ24} proposed a hybrid symbolic method to verify semantics preservation between original and transformed C/C++ programs targeting high-level synthesis, handling complex loop/buffer restructuring, which extends certified compilation principles to hardware synthesis pipelines.

\subsubsection*{Verification of constraint solvers.} 
There are several efforts developing formally verified constraint solvers, most of which focus on SAT solvers, with different methods though. For instance,  
Lescuyer and Sylvain 
use Rocq~\cite{lescuyer2008reflexive};
Shankar and Vaucher 
use PVS~\cite{ShankarV11};
Maric 
takes the shallow embedding in Isabelle/HOL~\cite{Maric10,NipkowPW02};
Fleury et al. 
use a refinement framework~\cite{BlanchetteFLW18,FleuryBL18,Fleury19};
Oe et al. 
utilize the verified programming language Guru~\cite{OeSOC12,StumpDPSS09}.
Besides SAT solvers, solvers~\cite{CarlierDG12,KanLRS22} for constraint programming over finite domains and string constraints 
are developed and verified with Rocq and Isabelle/HOL, respectively.
In contrast, our work focuses on a verified solver for a class of linear integer constraints, which are distant from the prior work.

\subsubsection*{Solving linear constraints.} The use of Presburger arithmetic and its fragments plays a pivotal role in verification and optimization. 
The FIRWINE constraints studied in this paper are a fragment of integer linear arithmetic, whose solving algorithms have been studied extensively  
\cite{2008Decision}. 
%
Modern SMT solvers such as Z3 \cite{LeonardoN08,LeonardoN07} and CVC5 \cite{2022cvc5} implement highly optimized decision procedures for quantifier-free linear integer arithmetic. 
Their Optimization Modulo Theories (OMT) extensions \cite{NokolajAL15,NiewenhuisO06} optimise objective functions over ILA constraints. While SMT solvers with OMT extensions can solve FIRWINE constraints efficiently, their use of e.g., internal heuristic strategies are too complicated to accommodate full verification.

\revise{ILP is NP-complete in general, and therefore admits exponential worst-case time complexity. While general ILP solvers (e.g., Gurobi and Z3) feature various heuristic optimizations, their implementations are unverified and hard to be verified}. 
In contrast, we provide a tailored procedure for solving FIRWINE constraints. It features a considerably lower complexity. Moreover, it is based on classical branch-and-bound (BaB) algorithms which are widely adopted in combinatorial optimization as well as the classical Floyd--Warshall algorithm for solving the all-pairs shortest path problem. \revise{We have empirically compared our implementation with the commercial ILP solver Gurobi on both manually constructed and real-world benchmarks. The results show that our implementation outperforms Gurobi or 
achieves comparable performance.} In particular, 
our solving algorithm enjoys a modular Rocq proof, which significantly pushes the frontier of a verified compiler for {\firrtl}. 




\section{Conclusion}\label{sec-conclusion}


In this paper, we have studied the width inference problem for {\firrtl} program compilation.
We derived a subclass of linear integer constraints from a {\firrtl} program where the bit widths of some components are not explicitly given, and  provided a complete procedure for solving these constraints. 
Furthermore, the procedure is formally verified in Rocq, where an  
OCaml implementation is also derived with competitive performance as the official InferWidths pass in {\firtool}. 

This work presents the first formally verified InferWidths pass, marking a significant step toward a production-quality, verified FIRRTL compiler. 
Broadly speaking, it advances the goal of formally verified hardware, 
which 
would contribute to a verified execution stack that preserves correctness guarantees from software down to hardware. 

As the next step, we will extend our approach to deal with dynamic shift left, which introduces exponential terms in the width constraints. Moreover, it is interesting to  implement the procedure proposed in this work in C++, which could replace the current InferWidths pass in {\firtool}. 

\clearpage

\section*{Data Availability Statement}
\revise{The Rocq implementation of our procedure, its  correctness proof, the source code of our width inference tool and benchmarks 
are available at~\cite{BFWInferWidths}.}

\bibliographystyle{splncs04}

\bibliography{firrtl}


\clearpage


\appendix 



\section{Syntax of \firrtl}\label{appendix:syntax}

$$\Phi_W := \bigwedge_{x_i\in X}\left(x_i\geq\min(t_{i,1},\ldots,t_{i,k_i})\right)$$

$$\begin{aligned}
\Phi_{\mathrm{nf}}:=\bigvee_{j_1\in[\ell]}\bigwedge_{j_2\in[m_{j_1}]}z_{j_1,j_2}\geq t_{j_1,j_2}
\end{aligned}$$

The syntactic rules of {\firrtl} are given in Figure~\ref{fig-firrtl-synt-simp} in Extended Backus-Naur Form, where \{$\cdots$\} denotes repetition (repeating zero or more times),
[$\cdots$] denotes option (either included or discarded),  and the terminals are doubly quoted by \mbox{\tt ``} and \mbox{\tt ''}.
The syntactic rules of {\firrtl}  presented here are a simplification of those in the official specification.\footnote{\url{https://github.com/chipsalliance/firrtl-spec/blob/main/spec.md}}
For instance, indent, dedent and newline terminals are omitted for readability;
the types and constructs for analog signals are also omitted, since we focus on digital circuits in this work. In the syntax, memory is also omitted since its semantics is still under active debate.\footnote{\url{https://github.com/chipsalliance/chisel/issues/3112}}


\begin{figure}[hp]
\vspace{-5mm}
{ 
\[
\begin{array}{l l l}
c & ::= & \mbox{``\fnprettyll{circuit}''}\ id\ \mbox{``:''} \ \{m\}\\
%
m & ::= & \mbox{``\fnprettyll{module}''} \ id \{p\} \{s\}\\
%
%
p & ::= & (\mbox{``\fnprettyll{input}''} \mid \mbox{``\fnprettyll{output}''})\ id\ \mbox{``:''}\ t \\
%
 s & ::= & \mbox{``\fnprettyll{wire}''}\ id\ \mbox{``:''}\ t  \mid \mbox{``\fnprettyll{reg}''}\ id\ \mbox{``:''}\ t\ e [\mbox{``(\fnprettyll{with}:\{\fnprettyll{reset}{\tt =>}(''} e\ \mbox{``,''}\ e \mbox{``)\})''}] \\
  &   &  \mid
  \mbox{``\fnprettyll{inst}''}\ id\ \mbox{``\fnprettyll{of}''}\ id\ \mid \mbox{``\fnprettyll{node}''}\ id\ \mbox{``{\tt =}''} id\  \mid r \mbox{``{\tt <=}''} e  \\
  &  & \mid r\ \mbox{``\fnprettyll{is invalid}''}  \mid \mbox{``\fnprettyll{when}''}\ e\ \mbox{``:''}\ \{s\} [\mbox{``\fnprettyll{else}:''} \{s\}] \mid \mbox{``\fnprettyll{skip}''}\\
%
%
%
  e & ::= & (\mbox{``\fnprettyll{UInt}''} | \mbox{``\fnprettyll{SInt}''}) [w] \mbox{``(''}int\mbox{``)''} \mid r \mid \mbox{``\fnprettyll{mux}(''}e\mbox{``,''} e\mbox{``,''} e\mbox{``)''}
        \mid op\\
%
%
r & ::= & id \mid r\mbox{``.''}id \mid r\mbox{``[''}int\mbox{``]''}\\
w & ::= & \mbox{``{\tt <}''} int \mbox{``{\tt >}''}\\
%
%
%
t & ::= & tg \mid ta \\
 tg & ::= &  \mbox{``\fnprettyll{Clock}''} \mid \mbox{``\fnprettyll{Reset}''} \mid \mbox{``\fnprettyll{AsyncReset}''} \mid (\mbox{``\fnprettyll{UInt}''} \mid \mbox{``\fnprettyll{SInt}''}) [w] \\
%
  ta & ::= &\mbox{``\{''} f \{f\}\mbox{``\}''} \mid t \mbox{``[''} int \mbox{``]''} \\
    f & ::= & [\mbox{``\fnprettyll{flip}''}]\ id\ \mbox{``:''}\ t\\
%
id & ::= & ( \mbox{``\_''} \mid l) \{ \mbox{``\_''} \mid l \mid d\} \\
op & ::= & op\_2e \mid op\_1e \mid op\_1e1i \mid op\_1e2i \\
op\_2e & ::= & op\_2e\_key \mbox{``(''}  e \mbox{``,''}  e \mbox{``)''} \\
%
op\_2e\_key & ::= &  \mbox{``\fnprettyll{add}''}  \mid  \mbox{``\fnprettyll{sub}''}  \mid \mbox{``\fnprettyll{mul}''} \mid \mbox{``\fnprettyll{div}''} \mid \mbox{``\fnprettyll{mod}''} \mid \mbox{``\fnprettyll{lt}''} \mid \mbox{``\fnprettyll{leq}''} \mid \mbox{``\fnprettyll{gt}''} \mid \mbox{``\fnprettyll{geq}''}\\
& &   \mid \mbox{``\fnprettyll{eq}''} \mid \mbox{``\fnprettyll{neq}''}  \mid \mbox{``\fnprettyll{dshl}''}  \mid \mbox{``\fnprettyll{dshr}''} \mid \mbox{``\fnprettyll{and}''} \mid \mbox{``\fnprettyll{or}''} \mid \mbox{``\fnprettyll{xor}''} \mid \mbox{``\fnprettyll{cat}''} \\
 %
op\_1e & ::= & op\_1e\_key \mbox{``(''}  e \mbox{``)''} \\
op\_1e\_key & ::= & \mbox{``\fnprettyll{asUInt}''} \mid \mbox{``\fnprettyll{asSInt}''} \mid \mbox{``\fnprettyll{asClock}''} \mid \mbox{``\fnprettyll{cvt}''} \mid \mbox{``\fnprettyll{neg}''}\\
& & \mid \mbox{``\fnprettyll{not}''}  \mid \mbox{``\fnprettyll{andr}''} \mid \mbox{``\fnprettyll{orr}''} \mid  \mbox{``\fnprettyll{xorr}''}\\
%
op\_1e1i & ::= & op\_1e1i\_key \mbox{``(''}  e \mbox{``,''} int \mbox{``)''} \\
op\_1e1i\_key & ::= & \mbox{``\fnprettyll{pad}''} \mid \mbox{``\fnprettyll{shl}''} \mid \mbox{``\fnprettyll{shr}''} \mid \mbox{``\fnprettyll{head}''} \mid \mbox{``\fnprettyll{tail}''} \\
op\_1e2i & ::= & op\_1e2i\_key \mbox{``(''}  e \mbox{``,''}  int  \mbox{``,''}  int \mbox{``)''} \\
op\_1e2i\_key & ::= & \mbox{``\fnprettyll{bits}''}\\
int & ::= & [\mbox{``-''}] d \{d\} \mid \mbox{``{\tt "}b''} [\mbox{``-''}] d\_b \{d\_b\} \mbox{``{\tt "}''} \\
& & \mid \mbox{``{\tt "}o''} [\mbox{``-''}] d\_o \{d\_o\} \mbox{``{\tt "}''} \mid
\mbox{``{\tt "}h''} [\mbox{``-''}] d\_h \{d\_h\} \mbox{``{\tt "}''} \\
%
%

%
\end{array}
\]
$d$ (resp. $d\_b$, $d\_o$ and $d\_h$) denotes a decimal (resp. binary, octal and hexadecimal) digit, and $l$ denotes a letter
}
\caption{Syntax of {\firrtl} }
\label{fig-firrtl-synt-simp}
\vspace{-4mm}
\end{figure}

According to the syntax, {\firrtl} features high-level constructs such as vector types, bundle types, conditional statements, connects, and modules.
These high-level constructs are then gradually removed by a sequence of \emph{lowering transformations}. During each lowering transformation, the circuit is rewritten into an equivalent circuit using simpler, lower-level
constructs. Eventually, the circuit is simplified to its most restricted form, resembling a structured netlist,
which allows for easy translation to an output language (e.g., Verilog).
This form is given the name lowered {\firrtl}  (i.e., \lofirrtl) and is a strict subset of the full {\firrtl} language.
For clarity, the original form is referred to as high {\firrtl}  (\hifirrtl).

A {\firrtl}  circuit is a {\lofirrtl} circuit if it obeys the following restrictions:
\begin{itemize}
\item All widths are explicitly defined.
\item No conditional statements (\prettyll{when}) are used.
\item All circuit components are declared with a ground type.
\item Each component is connected to at most once.
\end{itemize}
More precisely, a {\lofirrtl} circuit is a {\firrtl}  circuit that is defined by the syntactic rules in Figure~\ref{fig-lofirrtl-synt}
and satisfies that in each module, each component is connected to at most once.

\begin{figure}[htp]
\vspace{-3mm}
{
\[
\begin{array}{l l l}
 s & ::= & \mbox{``\fnprettyll{wire}''}\ id\ \mbox{``:''}\ t  \mid \mbox{``\fnprettyll{reg}''}\ id\ \mbox{``:''}\ t\ e [\mbox{``(\fnprettyll{with}:\{\fnprettyll{reset}{\tt =>}(''} e\mbox{``,''} e \mbox{``)\})''}] \\
  &   &  \mid
  \mbox{``\fnprettyll{inst}''}\ id\ \mbox{``\fnprettyll{of}''}\ id \mid \mbox{``\fnprettyll{node}''}\ id \mbox{``{\tt =}''} id  \mid r \mbox{``{\tt <=}''} e  \mid r\ \mbox{``\fnprettyll{is invalid}''}\mid \mbox{``\fnprettyll{skip}''}\\
%
e & ::= & (\mbox{``\fnprettyll{UInt}''} | \mbox{``\fnprettyll{SInt}''})\ w \mbox{``(''}int\mbox{``)''} \mid r \mid \mbox{``\fnprettyll{mux}(''}e\mbox{``,''} e\mbox{``,''} e\mbox{``)''}
 \mid op\\
%
%
r & ::= & id \mid r\mbox{``.''}id \mid r\mbox{``[''}int\mbox{``]''}\\
w & ::= & \mbox{``{\tt <}''} int \mbox{``{\tt >}''}\\
%
%
%
t & ::= & tg \\
 tg & ::= &  \mbox{``\fnprettyll{Clock}''} \mid \mbox{``\fnprettyll{AsyncReset}''} \mid (\mbox{``\fnprettyll{UInt}''} \mid \mbox{``\fnprettyll{SInt}''})\ w
%
%
\end{array}
\]
}
\caption{Syntax of {\lofirrtl} (Rules that are identical to those in Figure~\ref{fig-firrtl-synt-simp} are not repeated)}
\label{fig-lofirrtl-synt}
\end{figure}

\begin{figure}[htp]
\vspace{-2mm}
\centering
\begin{subfigure}[b]{0.4\textwidth}
         \centering
\begin{lstlisting}
module MyModule :
  input a: UInt<1>
  input b: UInt<1>
  output myport1: UInt<1>
  output myport2: UInt<1>
  myport1 <= b
  myport2 <= a
\end{lstlisting}
\vspace*{0.92\baselineskip}
\caption{}
\end{subfigure}
\hspace{4mm}
\begin{subfigure}[b]{0.4\textwidth}
         \centering
\begin{lstlisting}
module MyModule :
  input a: UInt<1>
  input b: UInt<1>
  output myport1: UInt<1>
  output myport2: UInt<1>
  myport1 <= a
  myport1 <= b
  myport2 <= a
\end{lstlisting}
\caption{}
\end{subfigure}
\vspace{-2mm}
\caption{Examples for illustrating the syntactic restrictions of {\lofirrtl}}
\label{fig-lofirrtl-exmp}
\vspace{-4mm}
\end{figure}

Note that in Figure~\ref{fig-lofirrtl-synt}, the rules for $f$, $id$, $op$ are the same as those in Figure~\ref{fig-firrtl-synt-simp}  (and thus omitted). 
In a nutshell, the syntactic rules in Figure~\ref{fig-lofirrtl-synt} are obtained from those in Figure~\ref{fig-firrtl-synt-simp} by replacing $[w]$ with $w$ (so that the width becomes mandatory),
removing the \prettyll{when} statement from the rule for $s$,
and removing the rule for aggregate types $ta$.
The requirement that each component is connected to at most once
gives {\lofirrtl} a flavor similar to static single-assignment (SSA) form of software programs.
Programs in SSA form are those where each variable is assigned at most once and each variable is defined before it is used.
For instance, the program in Figure~\ref{fig-lofirrtl-exmp}(a) is a {\lofirrtl} program,
while the program in Figure~\ref{fig-lofirrtl-exmp}(b) is not,
since the output port \prettyll{myport1} is connected to twice from the input ports \prettyll{a} and \prettyll{b}, respectively.

\begin{table}[!h]
\centering
\caption{Widths $w_{\tt e}$ of expressions {\tt e}}
\label{tab:bit-width-of-expressions-1}
\setlength{\tabcolsep}{8pt}
\begin{tabular}{|l|l|l|}
\hline
\textbf{Operation} & \textbf{Bit width} & \textbf{Note} \\ \hline \hline
{\tt e} = \prettyll{UInt<}$n_1$\prettyll{>(}$n_2$\prettyll{)} & $w_{\tt e} = n_1$ &  $n_1,n_2\in \natnum$ \\ \hline
{\tt e} = \prettyll{SInt<}$n_1$\prettyll{>(}$n_2$\prettyll{)} & $w_{\tt e} = n_1$ &  $n_1\in \natnum,n_2\in \intnum$ \\ \hline
${\tt e} = r$ & $w_{\tt e} = w_r$ & \makecell[l]{ 
 $r$ is a reference, $w_r\in\{$\prettyll{UInt,SInt,}\\
\prettyll{AsyncReset,Reset,Clock}$\}$} \\ \hline
{\tt e} = \prettyll{mux(c,}$\tt e_1,e_2)$ & $w_{\tt e} = \max(w_{\sf e_1},w_{\sf e_2})$ & $t_{\tt e_1} = t_{\tt e_2} \in \{$\prettyll{UInt,SInt}$\}$ \\ \hline
\makecell[l]{ 
{\tt e} = \prettyll{add}$\tt (e_1,e_2)$ or\\
{\tt e} = \prettyll{sub}$\tt (e_1,e_2)$}
& $w_{\tt e} = \max(w_{\sf e_1},w_{\sf e_2}) + 1$ & $t_{\tt e_1} = t_{\tt e_2} \in \{$\prettyll{UInt,SInt}$\}$ \\ \hline
\makecell[l]{
{\tt e} = \prettyll{mul}$\tt (e_1,e_2)$ or\\ 
{\tt e} = \prettyll{cat}$\tt (e_1,e_2)$}
& $w_{\tt e} = w_{\sf e_1}+w_{\sf e_2}$ & $t_{\tt e_1} = t_{\tt e_2} \in \{$\prettyll{UInt,SInt}$\}$ \\ \hline
{\tt e} = \prettyll{div}$\tt (e_1,e_2)$ & $w_{\tt e} = w_{\sf e_1}$ & $t_{\tt e_1} = t_{\tt e_2} =$ \prettyll{UInt} \\ \hline
{\tt e} = \prettyll{div}$\tt (e_1,e_2)$ & $w_{\tt e} = w_{\sf e_1}+1$ & $t_{\tt e_1} = t_{\tt e_2} =$ \prettyll{SInt} \\ \hline
{\tt e} = \prettyll{rem}$\tt (e_1,e_2)$ & $w_{\tt e} = \min(w_{\sf e_1},w_{\sf e_2})$ & $t_{\tt e_1} = t_{\tt e_2} \in \{$\prettyll{UInt,SInt}$\}$ \\ \hline
{\tt e} = \prettyll{pad}${\tt (e_1},n)$ & $w_{\tt e} = \max(w_{\sf e_1},n)$ & $t_{\tt e_1}\in \{$\prettyll{UInt,SInt}$\},n \geqslant 0$ \\ \hline

\makecell[l]{
  {\tt e} = \prettyll{eq}$\tt (e_1,e_2)$ or\\
  {\tt e} = \prettyll{neq}$\tt (e_1,e_2)$ or \\ 
  {\tt e} = \prettyll{lt}$\tt (e_1,e_2)$ or \\
  {\tt e} = \prettyll{leq}$\tt (e_1,e_2)$ or \\ 
  {\tt e} = \prettyll{gt}$\tt (e_1,e_2)$ or \\
  {\tt e} = \prettyll{geq}$\tt (e_1,e_2)$}
  & $w_{\tt e} = 1$  & $t_{\tt e_1} = t_{\tt e_2} \in \{$\prettyll{UInt,SInt}$\}$ \\
  \hline

{\tt e} = \prettyll{shl}${\tt (e_1,}n)$ & $w_{\tt e} = w_{\sf e_1} + n$ & $t_{\tt e_1}\in \{$\prettyll{UInt,SInt}$\},n \geqslant 0$ \\ \hline
{\tt e} = \prettyll{shr}${\tt (e_1,}n)$ & $w_{\tt e} = \max(w_{\sf e_1} -n,0)$ & $t_{\tt e_1}=$ \prettyll{UInt}, $n \geqslant 0$ \\ \hline
{\tt e} = \prettyll{shr}${\tt (e_1,}n)$ & $w_{\tt e} = \max(w_{\sf e_1} -n,1)$ & $t_{\tt e_1}=$ \prettyll{SInt}, $n \geqslant 0$ \\ \hline
{\tt e} = \prettyll{dshl}$\tt (e_1,e_2)$ & $w_{\tt e} = w_{\tt e_1} + 2^{w_{\tt e_2}} -1$ & $t_{\tt e_1}\in \{$\prettyll{UInt,SInt}$\},t_{\tt e_2}=$ \prettyll{UInt} \\ \hline
\makecell[l]{
{\tt e} = \prettyll{dshr}$\tt (e_1,e_2)$ or\\
{\tt e} = \prettyll{not}$\tt (e_1)$}
& $w_{\tt e} = w_{\tt e_1}$ & $t_{\tt e_1}\in \{$\prettyll{UInt,SInt}$\},t_{\tt e_2}=$ \prettyll{UInt} \\ \hline
{\tt e} = \prettyll{cvt}${\tt (e_1)}$ & $w_{\tt e} = w_{\sf e_1}+1$ & $t_{\tt e_1}=$ \prettyll{UInt} \\ \hline
{\tt e} = \prettyll{cvt}${\tt (e_1)}$ & $w_{\tt e} = w_{\sf e_1}$ & $t_{\tt e_1}=$ \prettyll{SInt} \\ \hline
{\tt e} = \prettyll{neg}$\tt (e_1)$ & $w_{\tt e} = w_{\tt e_1} + 1$ & $t_{\tt e_1}\in \{$\prettyll{UInt,SInt}$\}$ \\ \hline

\makecell[l]{
  {\tt e} = \prettyll{and}$\tt (e_1,e_2)$ or\\
  {\tt e} = \prettyll{or}$\tt (e_1,e_2)$ or \\ 
  {\tt e} = \prettyll{xor}$\tt (e_1,e_2)$}  
  & $w_{\tt e} = \max(w_{\sf e_1},w_{\sf e_2})$ 
  & $t_{\tt e_1} = t_{\tt e_2} \in \{$\prettyll{UInt,SInt}$\}$ \\  
  \hline

\makecell[l]{
  {\tt e} = \prettyll{andr}$\tt (e_1,e_2)$ or \\
  {\tt e} = \prettyll{orr}$\tt (e_1,e_2)$ or \\ 
  {\tt e} = \prettyll{xorr}$\tt (e_1,e_2)$}
  & $w_{\tt e} = 1$ & $t_{\tt e_1}\in \{$\prettyll{UInt,SInt}$\}$ \\   
  \hline
  
{\tt e} = \prettyll{bits}${\tt (e_1},hi,lo)$ & $w_{\tt e} = hi-lo+1$ & $t_{\tt e_1}\in \{$\prettyll{UInt,SInt}$\},0 \leqslant lo\le hi < w_{\tt e_1}$ \\ \hline
{\tt e} = \prettyll{head}${\tt (e_1},n)$ & $w_{\tt e} = n$ & $t_{\tt e_1}\in \{$\prettyll{UInt,SInt}$\},0 \leqslant n \leqslant w_{\tt e_1}$ \\ \hline
{\tt e} = \prettyll{tail}${\tt (e_1},n)$ & $w_{\tt e} = w_{\sf e_1}-n$ & $t_{\tt e_1}\in \{$\prettyll{UInt,SInt}$\},0 \leqslant n \leqslant w_{\tt e_1}$ \\ \hline\hline

\makecell[l]{
{\tt e} = \prettyll{asUInt}$\tt (e_1)$ or \\
{\tt e} = \prettyll{asSInt}$\tt (e_1)$}
& $w_{\tt e} = w_{\sf e_1}$ & $t_{\tt e_1} \in \{$\prettyll{UInt,SInt}$\}$ \\ \hline

\makecell[l]{
{\tt e} = \prettyll{asUInt}$\tt (e_1)$ or\\
{\tt e} = \prettyll{asSInt}$\tt (e_1)$}
  & $w_{\tt e} = 1$ 
  & \makecell[l]{
  $t_{\tt e_1}\in\{\text{\prettyll{AsyncReset,Reset,Clock}}\}$}  \\
  \hline
\end{tabular}
\end{table}

\section{Widths of expressions in FIRRTL specification}\label{sec:widthofExpr}
  
Given an expression ${\tt e}$, the width 
$w_{\tt e}$ of ${\tt e}$ is inductively defined in  Table~\ref{tab:bit-width-of-expressions-1}, where
$t_{\sf e}\in  \{\text{\prettyll{UInt,SInt,Reset,AsyncReset,Clock}}\}$ denotes the type of ${\tt e}$.
We note that in general, the reference $r$ and 
the expression ${\tt e}_i$ for $i=1,2$ used in the operator
\fnprettyll{mux} may be components of aggregate types (i.e., bundle types and vector types). The type checking ensures that the component connected from
$r$ or \fnprettyll{mux} has a compatible type. Thus, 
their connect statements can be split in a field- or element-wise fashion. As a result, only the widths of expressions of ground types are required.

\clearpage

\section{Further experimental results}\label{app:exp}


Here we list the running times per instance, where
the last three benchmarks belong to the REALWORLD benchmark suite while all the others belong to the MANUAL benchmark suite.

\smallskip

\scalebox{1}{
\begin{tabular}{|c|c|c|c|c|}
  \hline
  \multirow{2}{*}{\textbf{Benchmark}} & 
  \multirow{2}{*}{\textbf{\#cpnts}} & 
  \multicolumn{3}{c|}{\textbf{Time (ms) per instance}}\\ 
  \cline{3-5} 
  & & \textbf{firtool} & \textbf{Gurobi} & \textbf{BFWInferWidths}\\ 
  \hline \hline
Logic & 25 & 2.00 & 31.44 & \textbf{0.12} \\ \hline
AddNot & 1 & 2.00 & 9.93 & \textbf{0.01} \\ \hline
CombElseWhen & 1 & 1.90 & 11.87 & \textbf{0.01} \\ \hline
Conditional & 3 & 2.20 & 13.78 & \textbf{0.02} \\ \hline
DrawMux6 & 7 & 1.50 & 20.99 & \textbf{0.03} \\ \hline
ParamFunc & 1 & 1.10 & 13.05 & \textbf{0.01} \\ \hline
Timer & 5 & 0.90 & 10.08 & \textbf{0.02} \\ \hline
width-spec & 14 & 2.60 & 12.06 & \textbf{0.03} \\ \hline
CombOther & 1 & 2.90 & 7.41 & \textbf{0.01} \\ \hline
CombWhen & 1 & 1.80 & 9.56 & \textbf{0.01} \\ \hline
CombWireDefault & 1 & 0.70 & 16.97 & \textbf{0.01} \\ \hline
spec.md-code-example-019 & 1 & 1.50 & 5.31 & \textbf{0.01} \\ \hline
spec.md-code-example-020 & 1 & 0.40 & 6.12 & \textbf{0.01} \\ \hline
spec.md-code-example-104 & 5 & 0.30 & 4.78 & \textbf{0.01} \\ \hline
spec.md-code-example-106 & 5 & 2.70 & 16.26 & \textbf{0.02} \\ \hline
1088 & 1 & 0.30 & 5.08 & \textbf{0.02} \\ \hline
1089 & 1 & 0.40 & 5.17 & \textbf{0.01} \\ \hline
1090\_1 & 1 & 0.30 & 9.26 & \textbf{0.01} \\ \hline
1090\_2 & 1 & 0.30 & 8.28 & \textbf{0.01} \\ \hline
1091 & 1 & 0.40 & 4.87 & \textbf{0.01} \\ \hline
1092 & 1 & 0.40 & 4.52 & \textbf{0.02} \\ \hline
1093 & 1 & 0.40 & 4.39 & \textbf{0.01} \\ \hline
1094 & 1 & 0.40 & 4.42 & \textbf{0.01} \\ \hline
1095 & 1 & 0.40 & 4.24 & \textbf{0.01} \\ \hline
1096 & 1 & 0.40 & 4.26 & \textbf{0.01} \\ \hline
1097 & 1 & 0.40 & 4.55 & \textbf{0.01} \\ \hline
1098 & 1 & 0.30 & 5.21 & \textbf{0.01} \\ \hline
1099 & 1 & 0.40 & 4.36 & \textbf{0.01} \\ \hline
1108 & 1 & 0.30 & 4.35 & \textbf{0.01} \\ \hline
1110 & 1 & 1.00 & 49.53 & \textbf{0.01} \\ \hline
1112 & 1 & 0.90 & 5.63 & \textbf{0.01} \\ \hline
1113 & 1 & 0.40 & 3.60 & \textbf{0.01} \\ \hline
1115 & 1 & 0.40 & 4.63 & \textbf{0.01} \\ \hline
1116 & 1 & 0.30 & 4.93 & \textbf{0.01} \\ \hline
1118\_1 & 1 & 1.60 & 9.99 & \textbf{0.01} \\ \hline
1118\_2 & 1 & 0.30 & 10.68 & \textbf{0.01} \\ \hline
1271 & 3 & 1.10 & 4.42 & \textbf{0.05} \\ \hline
\end{tabular}}

\noindent
\scalebox{1}{
\begin{tabular}{|c|c|c|c|c|}
  \hline
  \multirow{2}{*}{\textbf{Benchmark}} & 
  \multirow{2}{*}{\textbf{\# cpnts}} & 
  \multicolumn{3}{c|}{\textbf{Time (ms) per instance}}\\ 
  \cline{3-5} 
  & & \textbf{firtool} & \textbf{Gurobi} & \textbf{BFWInferWidths}\\ 
  \hline \hline
1491 & 1 & 0.30 & 4.26 & \textbf{0.01} \\ \hline
2538 & 2 & 0.10 & 4.29 & \textbf{0.01} \\ \hline
5391 & 4 & 0.60 & 4.63 & \textbf{0.02} \\ \hline
DownCounter & 5 & 0.90 & 10.71 & \textbf{0.01} \\ \hline
DownTicker & 5 & 0.70 & 7.60 & \textbf{0.05} \\ \hline
Flasher & 24 & 0.70 & 8.56 & \textbf{0.12} \\ \hline
Flasher2 & 22 & 3.80 & 10.99 & \textbf{0.15} \\ \hline
GCD & 10 & 1.30 & 11.85 & \textbf{0.19} \\ \hline
OverflowTypeCircuit & 12 & 1.80 & 8.00 & \textbf{0.04} \\ \hline
Registers & 2 & 0.50 & 9.89 & \textbf{0.01} \\ \hline
Sequential & 4 & 1.80 & 10.57 & \textbf{0.02} \\ \hline
\makecell[c]{ShouldBeBadUInt\\SubtractWithGrow} & 4 & 1.10 & 8.24 & \textbf{0.02} \\ \hline
test & 1 & 1.50 & 10.23 & \textbf{0.01} \\ \hline
GCD 2 & 10 & 2.10 & 5.89 & \textbf{0.18} \\ \hline
SyncReset & 5 & 1.10 & 7.61 & \textbf{0.02} \\ \hline
MultiClockSubModuleTest & 18 & 1.70 & 9.53 & \textbf{0.05} \\ \hline
VcdAdder & 1 & 1.30 & 8.03 & \textbf{0.01} \\ \hline
ForwardingMemory & 9 & 1.80 & 8.46 & \textbf{0.04} \\ \hline
TrueDualPortMemory & 10 & 3.20 & 9.39 & \textbf{0.03} \\ \hline
lower-memories & 6 & 3.00 & 9.35 & \textbf{0.02} \\ \hline
Rob & 2,775 & 259.20 & 98.41 & \textbf{12.05} \\ \hline
CoreTester & 1,984 & 41.90 & 89.32 & \textbf{8.12} \\ \hline
CoreSoc & 3,315 & 83.20 & 85.70 & \textbf{44.89} \\ \hline
designed & 4 & / & 10.53 & \textbf{0.04} \\ \hline
designed2 & 4 & / & 7.24 & \textbf{0.07} \\ \hline
designed3 & 4 & / & 6.15 & \textbf{0.09} \\ \hline
designed4 & 3 & / & 4.47 & \textbf{0.08} \\ \hline
designed5 & 3 & / & 4.61 & \textbf{0.07} \\ \hline
designed6 & 3 & / & 4.64 & \textbf{0.03} \\ \hline
designed7 & 4 & / & 5.03 & \textbf{0.06} \\ \hline
designed8 & 4 & / & 6.47 & \textbf{0.09} \\ \hline
designed9 & 5 & / & 5.79 & \textbf{3.98} \\ \hline
designed10 & 4 & / & 4.38 & \textbf{0.06} \\ \hline
designed11 & 5 & / & 3.98 & \textbf{0.10} \\ \hline
6742 & 192 & / & 8.32 & \textbf{1.06} \\ \hline \hline
NutShell & 7,152 & 190.70 & 194.55 & \textbf{158.31} \\ \hline
Rocket Chip & 4,882 & 127.90 & 120.64 & \textbf{22.24} \\ \hline
RISC-V BOOM & 205,608 & 8,338.30 & \textbf{3,326.94} & 3,467.80 \\ \hline \hline

  \textbf{Average} & \textbf{2,976.43} & \textbf{144.54} & \textbf{59.41} & \textbf{48.96} \\ 
  \hline
\end{tabular}}


\end{document}